\documentclass[ecta,nameyear,draft]{econsocart}

\usepackage{natbib}
\usepackage[colorlinks,citecolor=blue,linkcolor=blue,urlcolor=blue,pagebackref]{hyperref}

\usepackage[nameinlink,capitalise,noabbrev]{cleveref}

\startlocaldefs
\usepackage[utf8]{inputenc}
\usepackage[toc,page]{appendix}
\usepackage[english]{babel}

\usepackage{amsmath, amssymb, amsthm, mathabx, bbm, graphicx, setspace, enumitem, tikz, subcaption,adjustbox,comment}
\usetikzlibrary{calc,intersections,arrows.meta,decorations.markings}
\usetikzlibrary{positioning}

\usepackage{booktabs}
\usepackage{float}
\usepackage[T1]{fontenc}
\usepackage{rotating}

\usepackage[dvipsnames]{xcolor}

\usepackage{ragged2e}

\usepackage{array}

\newcolumntype{L}[1]{>{\raggedright\arraybackslash}m{#1}}
\newcolumntype{C}[1]{>{\centering\arraybackslash}m{#1}}
\newcommand{\cellw}{4.2cm}   
     
\newcommand{\rowH}{3.8cm}    

\newcommand{\FC}[1]{%
  \parbox[c][2.4cm][c]{6.4cm}{\centering #1}%
}
\newcommand{\FCr}[1]{%
  \parbox[c][2.4cm][c]{1.6cm}{\centering #1}%
}
\newcommand{\FCc}[1]{%
  \parbox[c][1cm][c]{6.4cm}{\centering #1}%
}
\newcommand{\FCx}[1]{%
  \parbox[c][1cm][c]{1.6cm}{\centering #1}%
}

\tikzset{
  vertex/.style={circle,fill=black,inner sep=1.5pt},
  edge/.style={thick},
  vlabel/.style={font=\footnotesize}
}

\newtheorem{proposition}{Proposition}
\newtheorem{definition}{Definition}
\newtheorem{example}{Example}

\newtheorem{theorem}{Theorem}
\newtheorem{lemma}{Lemma}
\newtheorem{corollary}{Corollary}
\newtheorem{assumption}{Assumption}

\crefname{theorem}{Theorem}{Theorems}
\Crefname{theorem}{Theorem}{Theorems}
\crefname{lemma}{Lemma}{Lemmas}
\Crefname{lemma}{Lemma}{Lemmas}
\crefname{corollary}{Corollary}{Corollaries}
\Crefname{corollary}{Corollary}{Corollaries}
\crefname{proposition}{Proposition}{Propositions}
\Crefname{proposition}{Proposition}{Propositions}
\crefname{definition}{Definition}{Definitions}
\Crefname{definition}{Definition}{Definitions}
\crefname{example}{Example}{Examples}
\Crefname{example}{Example}{Examples}
\crefname{assumption}{Assumption}{Assumptions}
\Crefname{assumption}{Assumption}{Assumptions}
\crefname{property}{Property}{Properties}
\Crefname{property}{Property}{Properties}

\makeatletter
\renewenvironment{abstract}
  {%
    \par\small
    \noindent\textbf{Abstract}\par
    \noindent\RaggedRight
  }
  {\par}
\makeatother

\endlocaldefs

\begin{document}
\begin{frontmatter}

\title{How to Ask for Belief Statistics without Distortion?}
\runtitle{How to Ask for Belief Statistics without Distortion?}

\begin{aug}
%
%
%
\author[id=au1,addressref={add1}]{\fnms{Yi-Chun}~\snm{Chen}\ead[label=e1]{yichun@nus.edu.sg}}
\author[id=au2,addressref={add22}]{\fnms{Ruoyu}~\snm{Wang}\ead[label=e2]{ruoyu.wang@u.nus.edu}}
\author[id=au3,addressref={add2}]{\fnms{Xinhan}~\snm{Zhang}\ead[label=e3]{zhangxinhan@swufe.edu.cn}}
\address[id=add1]{%
\orgdiv{Department of Economics and Risk Management Institute},
\orgname{National University of Singapore}}

\address[id=add2]{%
\orgdiv{Center for Intelligence Economic Science},
\orgname{Southwestern University of Finance and Economics}}

\address[id=add22]{%
\orgdiv{Department of Economics},
\orgname{National University of Singapore}}

\end{aug}

\begin{abstract}
Belief elicitation is ubiquitous in experiments but can distort behavior in the main tasks. We study when, and how, an experimenter can ask for a series of action-dependent belief statistics after a subject chooses an action, while incentivize truthful reports without distorting the subject's optimal action in the main experimental tasks. We first propose a novel mechanism called the Counterfactual Scoring Rule (CSR), which achieves such nondistortionary elicitation of any single belief statistic by decomposing it into supplemental action-independent statistics. In contrast, when eliciting a fixed set of belief statistics without such decomposition, we show that robust nondistortionary elicitation is achievable if and only if the questions satisfy a joint alignment condition with the task payoff. The necessity of joint alignment is established through a graph theoretical approach, while its sufficiency follows from invoking an adaptation of the Becker–DeGroot–Marschak mechanism. Our characterization applies to experiments with general task-payoff structures and belief elicitation questions.
\end{abstract}

\end{frontmatter}

\section{Introduction}\label{sec: intro}
Belief elicitation prevails in laboratory and field experiments. Subjects first choose actions in a few ``tasks” and are then asked to report beliefs about their own performance or some payoff-relevant outcomes. For instance, subjects may be asked to answer a set of MCQs (tasks), and next, to report questions such as their expected total score, the probability of achieving full marks, whether their performance falls into a certain range, or how their performance is ranked among fellow participants, and so on. It is well known, however, that incentives for belief elicitation may distort the subjects' choices in the experimental tasks, as they maximize the total payoff from both the tasks and the belief elicitation questions. For instance, when elicitation rewards proximity between their predicted outcomes and the realized outcomes, subjects may prefer choosing actions that reduce the variance of payoff, even if those actions are suboptimal in terms of their task payoffs. In other words, belief elicitation may induce distortion  (\citet{healy2025belief,blanco2010belief,niederle2007women,chambers2021dynamic,hu2024confidence,enke2023cognitive}).

This paper offers a complete characterization of nondistortionary belief elicitation in experiments. Our results apply to experiments with arbitrary task payoffs and elicitation questions, generalizing the characterization in \citet{pkeski2025nondistortionary} which is tailored to specific task payoffs and asking a single elicitation question. 
Formally, we model a subject who chooses an action $a$ to maximize a state-dependent task payoff $u(a, \theta)$ given her belief over the state $\theta$. After the action is chosen, the experimenter asks $m$ questions/statistics regarding her beliefs $X_j(a, \theta)$ with $j=1,...,m$ and the subject reports $r_j$ regarding the expectation $\mathbb{E}_p[X_j(a, \theta)]$ for each question $j=1,...,m$. A \textit{nondistortionary} elicitation scheme employs an aggregate payoff function defined on actions, states, and reports such that:\footnote{For example, subjects answer multiple-choice questions as well as guess their rank among all subjects, then receive monetary payments based on the accuracy of the multiple-choice questions and the correctness of their guess about the rank.} (i) for any given action, it is optimal to truthfully report all $m$ expected statistics; and (ii) assuming the truthful reporting in (i), the action that maximizes $\mathbb{E}_p[u(a, \theta)]$ also maximizes the decision maker's aggregate expected payoff.

Our first main result is that once we shift the focus from evaluating a single question in isolation to designing incentives for the set of desired belief statistics as a whole, \emph{any} belief statistic $Y$ is incentivizable by asking supplemental questions. We propose a novel belief elicitation mechanism called the counterfactual scoring rule (CSR), in which the experimenter asks, ``What is your expectation of $Y$ \textit{had you chosen $a$}?''---namely, the expectation of $Y(a,\theta)$ conditional on each action $a$ that the subject may have chosen in the task. The CSR calculates the sum of the squared losses of the reported statistics $r(a)$ against the realized $Y(a,\theta)$ over all actions. Because the CSR does not depend on which action the subject chooses in the task, it induces no distortion; moreover, adding any arbitrarily small proportion of $u(a,\theta)$ preserves a strict incentive to choose the optimal task action. More generally, the same principle applies to incentivizing any question $Y$ by performing a rank decomposition of $Y$ and asking the subject to report the expectation of each basis statistic of $Y$; see \cref{cor: m-decomposable}. Relative to eliciting the entire belief, the CSR substantially simplifies the elicitation questionnaire when there are fewer actions than states. Moreover, the CSR imposes no restriction on how the question payoff relates to the task payoff, contrary to the findings of \citet{pkeski2025nondistortionary} regarding the elicitation of a single belief statistic.

Next, we characterize incentivizing a given set of belief statistics. Our characterization builds on a novel condition that we call joint alignment. The condition requires that for each action taken, the belief elicitation questions are related to the task payoffs via an affine transformation with an action-independent shift. On one hand, joint alignment is strictly weaker than individual alignment in \citet{pkeski2025nondistortionary}, where the affine transformation is applied question-by-question; hence, it underscores the importance of designing the set of questions as a whole. On the other hand, joint alignment remains sufficient for nondistortionary belief elicitation. More importantly, joint alignment characterizes what we call \emph{robustly incentivizable} questions; that is, conditions (i) and (ii) still hold even when the questions are slightly misspecified or biasedly interpreted by the subject. 

To show that joint alignment is necessary for robust incentivizability, our argument builds on the adjacency graph proposed by \citet{pkeski2025nondistortionary}. Each edge of the adjacency graph connects pairs of actions that are the only two optimal actions under some belief. These edges mark locally fragile incentives: a small tilt can break the tie and distort the optimal task action. For each edge $(a, b)$, we define a matrix edge flow $w(a, b)$, interpreted as the local "exchange rates" across belief questions that preserve indifference between $a$ and $b$; off-diagonal entries capture cross-question spillovers. Our novel step is to apply Kirchhoff's (voltage) law: if the ordered product of edge flows around every cycle equals $I$ (equivalently, flows are path-independent), then $w(a, b)=\Gamma(b) \Gamma(a)^{-1}$ for some (node) potential $\Gamma(\cdot)$.\footnote{See \citet{godsil2013algebraic,monderer1996potential}.} The potential delivers the affine transformation underlying the joint alignment condition. The necessity argument applies to any graph, provided that the length of each elementary cycle in its minimum cycle basis is bounded by the ``slack'' between the state space size and the number of questions. 

The paper proceeds as follows. \cref{Section: BeliefElicitationProblem} describes the model and the counterfactual scoring rule. Section 3 introduces the joint alignment condition and establishes its sufficiency. \cref{sec: necessity} develops the adjacency-graph approach and proves necessity for one-block graphs, organizing the argument around edge flows and integrability. Section 5 extends the analysis to product structures and multi-block decompositions, yielding modular verification and construction results. Section 6 discusses conceptual and methodological links to \citet{pkeski2025nondistortionary} and illustrates applications to common experimental belief questions.

\section{Belief Elicitation Problem}\label{Section: BeliefElicitationProblem}

Consider a single-agent state-contingent task $(\Theta, A, u)$ (the ``task'') with a finite state space $\Theta$ and a finite action space $A$, and the state-contingent payoff vectors $u(a) \in \mathbb{R}^{\Theta}$ for each $a\in A$. The decision maker (DM) makes his choice of action $a$ to maximize the expected task payoff $\mathbb{E}_pu(a)$ under his belief $p$ over the state space. As $p$ is unobservable to the experimenter, she would like to elicit some action-dependent belief statistics $r=(r_1,\cdots,r_m) \in \mathbb{R}^m$ through a series of questions. She does so in the following way: after DM choosing an action $a$ in the task, the experimenter asks him to report the expectation $r_j$ of each question $X_j(a)$ in $X=(X_1,\cdots,X_m)$, where each question is an action-dependent vector $X_j(a)\in\mathbb{R}^\Theta$. DM will be awarded according to the elicitation method $V: \mathbb{R}^m \times A \times \Theta \rightarrow$ $[0,1]$ through report $r$, the action chosen $a$, and the true state $\theta$. Say the question profile $X$ is\textit{ incentivizable} if there exists $V$ such that $$\arg \max _{(r,a)} \mathbb{E}_p[V(r, a, \theta)]=\left\{\left(\mathbb{E}_p X(a ; \theta), a\right): a \in \arg \max _{a \in A} \sum_\theta p(\theta) u(a ; \theta)\right\}.$$
In words, there exists a payment scheme that elicits DM's true belief while keeping her incentive at the first stage undistorted. The setting is flexible in accommodating two interpretations. First, the experimenter could be interested in \textit{all} the $m$ statistics, so that entire question profile $X$ is taken as given and the number of questions $m$ fixed. Thus, the above setup generalizes \citet{pkeski2025nondistortionary} (henceforth, PS25) by extending the number of statistics of interest beyond 1. Alternatively, the experimenter might be interested in certain belief statistics (e.g., the first couple of statistics in $r$), but she is allowed to ask a number of supplemental questions (the remaining in $r$) for the purpose of eliciting the desired ones. 

\subsection{The Counterfactual Scoring Rule}
To illustrate the mileages of asking multiple questions, we assume that the experimenter is interested in eliciting only one statistic $Y$. Consider the following experimental recipe using the counterfactual scoring rule mechanism described in the introduction, which resembles the \textit{revelation principle} in mechanism design. Let there be in total $m$ possible actions $A=\{a_1,a_2,\dots,a_m\}$.
\begin{enumerate}
    \item DM chooses action $a_i\in A$ for the task payoff $u(a,\theta)$;
    \item DM reports $\tilde{r}_j=\mathbb{E}_pY(a_j)$ for each $a_i\in A$;
    \item DM is paid a fixed prize with probability proportional to $V^{CSR}({r},a_i,\theta)=-\sum_{j=1}^m\left[Y(a_j,\theta)-r_j\right]^2+u(a_i,\theta).$\footnote{We use a fixed prize lottery to avoid effect from risk attitudes.}
\end{enumerate}

\noindent 
In expectation, $\mathbb{E}_pV^{CSR}(r, a, \theta)=-\sum_{i=j}^m\left[\mathbb{E}_pY(a_j,\theta)- r_j\right]^2+\mathbb{E}_pu(a_i,\theta)$, so the optimal reports $r_i=\mathbb{E}_pY(a_i,\theta)$ while the optimal action coincides with that for the original task. Intuitively, this mechanism asks DM $m$ questions in total, each in the form ``What is the expectation of $Y$ under your belief \textit{if you had chosen $a_i$}?'' Thus, instead of trying to directly elicit the action $a$ simultaneously with the belief $p$, it symmetrizes over $a$ by traversing over all actions in $A$; thus, it eliminates the action dependence of $Y$ along with the distortion. When there are multiple questions $Y_i$ of interest, the experimenter can simply ask them one by one, each with a separate CSR. As it turns out, the CSR imposes no restriction on how the question payoff relates to the task payoff, contrary to the findings of \citet{pkeski2025nondistortionary}:

\begin{proposition}
Any question $Y$ is incentivizable via the CSR mechanism for any task $u$.
\end{proposition}

We illustrate CSR using a task carried out by \citet{niederle2007women} for testing gender difference in their decision whether to enter a competitive environment. After solving for a series of problems, each subject chooses between two compensation schemes: piece rate versus tournament. Intuitively, DM is paid according the number of correct answers in the former and according to his performance ranking within fellow participants. Let the binary action set $A=\{\text{PR},\text{T}\}$ denote the two compensation schemes, and $\Theta$ denote the set of all possible performance profiles across all participants. Now suppose the experimenter is interested in DM's estimated likelihood that his payoff $u(a,\theta)$ is no less than a certain amount $t$, i.e. $r(a)=\mathbb{E}_p\mathbf{1}\{u(a,\theta)\ge t\}$.\footnote{The authors originally ask for reporting DM's expected rank.} There may be action distortion due to hedging. Think about a subject faces three very intelligent opponents, she would expect to come in last place. In this situation, choosing the piece rate in time would yield a larger payoff $u(a,\theta)$, but choosing the tournament would provide a more certain likelihood—that she will definitely come in last place. Therefore, she has an incentive to distort her optimal action and then seek a larger and/or more certain reward in the belief elicitation problem.

The CSR mechanism amounts to asking the following \textit{two} questions: ``What is your estimated likelihood that your payoff  is no less than $t$ \textit{had you chosen the piece rate (resp. tournament) scheme}?'', and then evaluate both answers with the standard quadratic scoring rule. The key advantage of the CSR is that the belief-scoring part is invariant across the two choices of the subject, and thereby introducing no distortion. The subject effectively needs to assess a single underlying uncertainty—her rank—under one additional counterfactual choice, which is far from contemplating a full belief over the high-dimensional score profile $\theta$.\footnote{The same setup is replicated in \citet{mobius2022managing}, where they also further elicit a different belief statistic for a similar experiment.
}

The similar operational recipe can be applied to two other popular types of experimental tasks. In a \textit{lottery choice} task, DM is faced with two or more candidate lotteries $L_i$, each viewed as an subjective probability distribution on some finite support $\Theta$ as monetary prizes. An experimenter interested in beliefs statistics such as ``What is your estimated chance of receiving more than \$$x$?'' can apply the SR by asking for all $i$, ``What is your estimated chance of receiving more than \$x \textit{if you had chosen $L_i$}?'' Similarly, in an \textit{investment} task, DM decides whether to invest in a safer or riskier financial asset, such as a money market fund versus an index-based ETF, or picking among different stocks. To DM, the return of each asset can be viewed as a subjective distribution on some support $\Theta$. An experimenter interested in beliefs statistics such as ``What is your estimated chance of having a return between $y_1$\% and $y_2$\%?'' can apply CSR by simply asking two questions ``What is your estimated chance of having a return between $y_1$\% and $y_2$\% \textit{suppose you had chosen the safer (riskier) asset}?''

We briefly discuss the practical applicability of the CSR. 
\begin{enumerate}
    \item When the action space $A$ is small, the CSR is very efficient as it requires reporting $|A|$ statistics.
    \item When $|\Theta|\ll|A|$, i.e. the state space $\Theta$ is relatively small, one can use another mirroring belief revelation mechanism that instead asks for the entire belief $p$.
    Formally, let $V^{BR}(a,\theta,r)=\sum_{\theta_k\in\Theta}\left(\mathbf{1}\{\theta=\theta_k\}-r_k\right)^2+u(a_i,\theta)$, where $r_k$ is the reported belief $p(\theta_k)$. Notice that both mechanisms do not depend on the specific form of the task, and thus are \textit{universally} applicable in experimental setting. 
    \item When both $|A|$ and $|\Theta|$ are large, one might argue that the mechanisms above could lose applicability, as there could be a practical constraint on the number of statistics to be asked. We show by applying rank decomposition to $Y$ in \cref{subsec: joint vs individual} that it suffices to take the number of questions $m$ as the rank of the $Y(a,\theta)$ matrix and apply the CSR elicit the $m$ questions.
    \item Alternatively, the experimenter may be willing to allow certain degree of distortion in exchange of asking fewer questions. In particular, she may opt to partition $A=\bigcup_k A^k$ and only ask the conditional expectation of $Y(a)$ on the partition. The ``coarse'' CSR grants the experimenter flexibility to remove the action dependence of $Y$ partially, if the experimenter is confident of no distortion within each $A^k$.
\end{enumerate}

\section{Joint Alignment and Weak Alignment}\label{sec: sufficient}

To fully characterize incentivizability, we now consider incentivizing a given set of belief statistics.

\subsection{Joint Alignment: the First Interpretation}
We first introduce an alignment condition generalized from PS25. For notational ease, we view $X$ as an $m\times |\Theta|$ matrix, $u$ and each $X_j$ as $|\Theta|$-dimensional row vectors. Let $\mathbf{1}\in\mathbb{R}^{|\Theta|}$ stands for the row vector of all ones. Illustrating examples are provided in \cref{subsec: examples}.
\begin{definition}[Joint Alignment]
    An $m$-dimensional question profile $X$ is jointly aligned with $u$ if there exists a column vector $\lambda\in \mathbb{R}^m$, a matrix $d\in \mathbb{R}^{m\times|\Theta|}$, and for each $a\in A$, an invertible matrix $\Gamma(a) \in \mathbb{R}^{m \times m} $ and a column vector $\kappa(a) \in \mathbb{R}^m$, such that $X(a)=\Gamma(a)\big(\lambda u(a)+d\big)+\kappa(a)\mathbf{1}$ for all $a\in A$.
\end{definition}
A question $X_j$ is said to be individually aligned with $u$ if we replace $X$ with $X_j$ in the above definition. Without loss of generality, we can take $\lambda\in\{0,1\}$ in above by rescaling $\gamma(a)$ and $d$. We show below that a variation of the renowned BDM mechanism can incentivize a jointly aligned set of questions.
\begin{lemma}\label{lem: joint sufficiency}
    If question profile $X$ is jointly aligned with $u$, then it is incentivizable.
\end{lemma}
\textbf{A BDM mechanism.} We first describe the elicitation mechanism for the simplest case $X_0(a)= u(a)+d$. Let $L\in\mathbb{R}$ satisfy $L< u(a;\theta)+d(\theta)$ for all $a$ and $\theta$.\footnote{In the sense of the original BDM mechanism, $V(r,a,\theta)=\int_{L}^{r} \big[  u(a;\theta)+d(\theta)\big]\,dx+\int_{r}^{M} x\,dx-\frac{M^{2}}{2}$ for some $L,M\in\mathbb{R}$ satisfy $L< u(a;\theta)+d(\theta)<M$ for all $a$ and $\theta$.} Set $V_0(r,a,\theta)=\big[  u(a;\theta)+d(\theta)\big]\big(r-L\big)-\frac12 r^2$, so that the expectation $\mathbb{E}_p V_0(r,a,\theta)=-\frac{1}{2}r^2+r\mathbb{E}_p\big(  u(a;\theta)+d(\theta)\big)+\operatorname{constant}$. As in the action-revelation mechanism $V^{CSR}$, under belief $p$, the optimal choice of $r$ given action $a$ is $r^*(a)=\mathbb{E}_{p}\!\left[u(a;\theta)+d(\theta)\right]=\mathbb{E}_{p}\!\left[X_0(a;\theta)\right]$. The payoff then becomes $\mathbb{E}_p V(r^*(a),a,\theta)=\frac{1}{2}\big[\mathbb{E}_p\left(u(a;\theta)+d(\theta)\right)-L\big]^{2}$. Since $L<  u(a;\theta)+d(\theta)$, maximizing $V$ over $a\in A$ is equivalent to maximizing $\mathbb{E}_p u(a;\theta)$. Hence, the optimal choice of action $a$ is the same as in the original task with payoff $u$, and the optimal report is undistorted. 

When $X_0(a)=d$, we can simply let $V_0(r,a,\theta)=u(a;\theta)+d(\theta)(r-L)-\frac{1}{2}r^2$ as in the action-revelation mechanism $V^{CSR}$ and apply a similar argument.

For a general joint aligned $X_0$, we first transform the question profile into $\tilde{X}(a)=\Gamma(a)^{-1}\left[X(a)-\kappa(a)\mathbf{1}\right]$, and correspondingly the transformed report $\tilde{r}(a)=\Gamma(a)^{-1}\left[r(a)-\kappa(a)\right]$. For each individual question $j$ in the rows of the transformed question, the BDM mechanism $V_0(\tilde{r}_j,a,\theta)$ can separately elicit DM's expectation $\tilde{r}_j(a)$ of $\tilde{X}_j(a)$ for any fixed action $a$. Let $V(r,a,\theta)=\sum_jV_0(\tilde{r}_j,a,\theta)$. Since each question $\tilde{X}_j(a)$ is of the form $u(a)+d_j$, the undistorted incentive ensures that DM will choose the optimal action for the original task $u$, so that $\tilde{X}$ is jointly incentivizable.

\subsection{Incentivization through Supplemental Questions: the Second Interpretation}\label{subsec: joint vs individual}

We now return to the second interpretation and examine the added value of joint alignment beyond individual alignment from the perspective of experiment design. Again, assume for simplicity that the experimenter is interested in one belief statistics described by question $Y$.

\begin{definition}[$m-$incentivizability]
    Question $Y$ is $m-$incentivizable if there exists an incentivizable $m-$dimensional question profile $X$ such that $X_1=Y$.
\end{definition}
 \cref{lem: joint sufficiency} implies the following immediate sufficient condition for $m-$incentivizability.

\begin{definition}[Weak alignment]\label{def: weak align}
Question $Y$ is weakly aligned with $u$ if there exist real numbers $\{\lambda_k\}_{k=1}^m$, vectors $\{d_k(\theta)\}_{k=1}^m$, and $\{\gamma_k(a)\}_{k=1}^m$ for each $a\in A$, such that $Y(a,\theta)=\left(\sum_{k=1}^m \gamma_k(a)\lambda_k\right) u(a,\theta)+\left(\sum_{k=1}^m \gamma_k(a) d_k(\theta)\right)$.  Say $Y$ is $m$-aligned with $u$ if $m$ is the minimal number such that $Y$ can be expressed in the above form.
\end{definition}
Equivalently, $Y$ is the first question $X_1$ for a jointly aligned $X$. Notice that the $\kappa$ term is absorbed into the second bracket as there is only one question $Y$ in consideration. 
\begin{proposition}\label{prop: m-incen}
     Question $Y$ is $m-$incentivizable if it is $m$-aligned with $u$.
\end{proposition}
Though the equation in \cref{def: weak align} might look complicated at the first glance, it reduces to a much simpler take-home message if we remove the role play by $u(a)$. Say a question $Y$ is a \textit{$m$-decomposable} if there exist vectors $\{d_k(\theta)\}_{k=1}^m$, and $\{\gamma_k(a)\}_{k=1}^m$ for each $a\in A$, such that 
\begin{equation}\label{eq: separable}
    Y(a,\theta)=\sum_{k=1}^m \gamma_k(a) d_k(\theta);
\end{equation} in other words, $Y(a,\theta)$ can be written as sum of the products of at most $m$ pair of functions $\left(\gamma_k(a),d_k(\theta)\right)$.
\begin{corollary}\label{cor: m-decomposable}
An $m$-decomposable question is $m-$incentivizable regardless of task payoff $u$.
\end{corollary}
The above result can be viewed as a ``recipe'' for experiment design: for experimenters who would like to elicit belief statistics that are not individually incentivizable, they can now accompany them with some other properly designed ones so that the interdependence between these questions enables joint incentivization. Mathematically,  \cref{eq: separable} amounts to finding an $|A|\times m$ matrix $G$ and an $m\times |\Theta|$ matrix $D$, such that $Y=GD$, where $Y$ is viewed as an $|A|\times |\Theta|$ matrix with its $(i,k)$-th entry being $(a_i,\theta_k)$. By \cref{lem: matrix decomposition} in \cref{appendix: proof}, the minimally possible $m$ equals the rank of $Y$. Specifically, one may take the rows of $D$ as $m$ linearly independent rows of $Y$. Then, each of the supplemental questions $X_j$ are exactly the question $Y$ with some fixed prefixed action $a_i$. When $\operatorname{rank}Y\ll |A|$, it suffices to ask for $Y(a_i)$ for a small subset of actions $a_i\in A$.

For instance, consider the lottery choice task where DM is asked to choose one out of several lotteries with large supports, so that the number of states $|\Theta|$ is much larger than the number of choices $|A|$. Then, there exist solutions with $m=\operatorname{rank}(Y)\leq \min\{|A|,|
\Theta|\}\ll|\Theta|$.\footnote{In another extreme, when the question $Y$ has rank 1, there exists a solution with $m=1$, as now $Y$ is already individually aligned with $u$.} In fact, one can let $m=|A|$, $D=Y$ and $G=I$, so that $d_i(\theta)=Y(a_i,\theta)$. This would be optimal when $Y$ has full rank over $A$. The corresponding supplemented question profile $X(a_i)=\Gamma(a_i)Y$, where $\Gamma(a_i)=\sigma_i=(e_i,e_{i+1},\dots,e_{m},e_1,\dots,e_{i-1})$ is the $i$-th permutation matrix with $e_k$ representing the column unit vector with its $k$-th row equal to 1 and the rest being 0. Observe that $X_1(a_i)=Y(a_i)$ for each $i$. We can use the BDM mechanism after \cref{lem: joint sufficiency} to incentivize this $X$: let
\begin{align*}
     V(a_i,\theta,r)&=\sum_{j=1}^m\left[Y(a_j,\theta)(\sigma_i^{-1}r)_j-\frac{1}{2} (\sigma_i^{-1}r)_j^2\right]+u(a_i,\theta)\\
     &=\sum_{j=1}^m\left[Y(a_j,\theta)r_{j-i+1}-\frac{1}{2} r_{j-i+1}^2\right]+u(a_i,\theta)
\end{align*}
\noindent Then, to maximize $\mathbb{E}_pV(a,\theta,r)$, the sum necessitates truthful report $r(a)=\mathbb{E}_pX(a)$, and hence the desired $\mathbb{E}_pY(a)$, while the remaining $u$ term eliminates distortion. In fact, after a change of labels, this is exactly the action-revelation mechanism $V^{CSR}$.\footnote{As another example, if $Y(a,\theta)=\sum_{k=1}^mf(a)\theta^{k-1}$ is a polynomial of $\theta$, i.e., the desired belief statistic is an action-dependent moment of the state as a random variable, then we can elicit it with at most $m-1$ supplemental questions. In fact, each supplemental $X_j$ in the form of $X_j(a,\theta)=\theta^{j-1}$, so that these questions essentially ask for the first $m-1$ moments of $\theta$.}

One may further solve \cref{eq: separable}, or the full equation in \cref{def: weak align} for a ``first-best'' solution with a minimal number of reported statistics $r$, whereas the preceding can be viewed as a ``second-best'', but more practical and economically intuitive recipe. In fact, by the necessity results in the next section, these are the only recipes available. Hence, for an $m$-aligned question $Y$, at least $m-1$ supplemental questions need to be asked together to achieve joint alignment, so that the number $m$ represents the \textit{complexity} that must be involved for the sake of incentivizing truthful report.

\subsection{Examples}\label{subsec: examples}
\begin{example}[MCQ]\label{eg: individual not aligned}
    Consider an MCQ task. The action space $A$ coincides with the state space $\Theta$, with the true state $\theta\in\Theta$ represents the correct answer. The payoff $u(a,\theta)=\mathbf{1}\{a=\theta\}$, i.e. DM receives payoff 1 if she chooses the correct answer, and 0 otherwise.
    
    Let the state space $\Theta=\{0,\frac{1}{2},1\}$ being real numbers. Now suppose the experimenter instead wish to ask DM ``what is the average squared deviation from the true state?'' This amounts to setting $X_1(a,\theta)=(a-\theta)^2$, so that a truthful report $r_1(a)=\mathbb{E}_pX_1(a,\cdot)=\mathbb{E}_p(a-\theta)^2$ reveals the desired answer.
\end{example}
One can verify that $X_1$ is not aligned with $u$. In fact, $X_1$ is not incentivizable through a similar BDM mechanism. To see the distortion, apply the above BDM payment rule: for any fixed $a$ the optimal report is $r^*(a)=\mathbb{E}_p[X_1(a,\theta)]$, and substituting back gives $\mathbb{E}_p[V(r^*(a),a,\theta)]=\frac12\big(\mathbb{E}_p[X_1(a,\theta)]-L\big)^2$, so the induced action choice is $a\in\arg\max_{a\in A}\mathbb{E}_p[X_1(a,\theta)]$. Now take $p(0)=0.3$, $p(\tfrac12)=0.6$, $p(1)=0.1$; then $\mathbb{E}_p[X_1(0,\theta)]=0.25$, $\mathbb{E}_p[X_1(\tfrac12,\theta)]=0.10$, and $\mathbb{E}_p[X_1(1,\theta)]=0.45$, so BDM selects $a=1$, whereas the MCQ task payoff $u(a,\theta)=\mathbf{1}\{a=\theta\}$ is maximized at $a=\tfrac12$.

However, when the experimenter is allowed to ask multiple questions, the hedging problem can be attenuated through \textit{joint} incentivization with a properly designed question $X_2$.

\begin{example}[MCQ, continued]\label{eg: joint inc}
    Consider the same MCQ task with the state space $\Theta=\{0,\frac{1}{2},1\}$. Now suppose the experimenter further asks DM ``what is the second momentum of the true state?'' This amounts to setting $X_2(a,\theta)=\theta^2$, so that a truthful report $r_2(a)=\mathbb{E}_pX_2(a,\cdot)=\mathbb{E}_p\theta^2$ reveals the desired answer.

    The question profile $X=(X_1;X_2)$ is jointly aligned with $u$ by setting $\Gamma(a)=\begin{bmatrix}
    a & 1\\
    0 & 1
\end{bmatrix}$, $\lambda(\theta)=\begin{bmatrix}
    -2\theta \\
    \theta^2 
\end{bmatrix}$, and $\kappa(a)=\begin{bmatrix}
    a^2\\
    0
\end{bmatrix}$. By \cref{lem: joint sufficiency}, $X$ is jointly incentivizable through a modified BDM mechanism.
\end{example}

\subsection{Joint Alignment and Change of Basis}\label{subsec: joint and change of basis}
Here we discuss the added value of joint alignment beyond individual alignment from the theoretical perspective, in addition to the applied view in \cref{subsec: joint vs individual}. Suppose for simplicity that ${X}(a)=\Gamma(a) \left(\lambda{u}(a) +{d} \right)$ for some column vector $\lambda\in \mathbb{R}^m$, matrix $d\in \mathbb{R}^{m\times|\Theta|}$, and invertible matrices $\Gamma(a) \in \mathbb{R}^{m \times m} $ for all $a\in A$.

First, if $\Gamma(a)$ is a diagonal matrix for each $a\in A$, then each question $X_j$ is itself individually aligned with $u$, and thus joint incentivization is trivially reduced to the individual case.

Second, if $\Gamma(a_0)$ is not a diagonal matrix for some $a_0$, but is diagonalizable (via eigenvalue decomposition). Suppose for each $a\in A$, $\Gamma(a)$ is diagonalizable with the same orthogonal matrix $Q \in \mathbb{R}^{m \times m}$, i.e. $\Gamma(a)=Q^{-1}\tilde{\Gamma}(a)Q$ for some diagonal matrix $\tilde{\Gamma}(a)$.\footnote{This would be the case, for instance, when each $\Gamma_{ab}$ matrix in \cref{lemma: adjacency MQ} in \cref{subsec: adjacency} is a diagonal matrix. We maintain this condition for the case below for simplicity.} Joint alignment can be rewritten into 
\begin{equation*}
    \begin{split}
        {X}(a) &= Q^{-1}\tilde{\Gamma}(a)Q\left(\lambda{u}(a) +{d} \right)\\
        \Rightarrow  Q{X}(a) &=\tilde{\Gamma}(a)Q\left(\lambda{u}(a) +{d} \right)\\
        &= \tilde{\Gamma}(a)  \left(Q\lambda{u}(a) + Q{d}\right).
    \end{split}
\end{equation*}If we let $\tilde{X}=QX$, then the transformed questions $\tilde{X}$ is aligned with $u$ via the new column vector $\tilde{\lambda}=Q\lambda$, matrix $\tilde{d}=Qd$, and the diagonal matrices $\tilde{\Gamma}(a)$. In other words, we are back to the first case after a change of basis, and the transformed questions can now be individually incentivized.

Generally, if $\Gamma(a_0)$ is not diagonalizable, we can implement the Jordan decomposition and write $\Gamma(a_0)=Q^{-1}\hat{\Gamma}(a_0)Q$ where $\hat{\Gamma}(a_0)$ is a Jordan normal form: it has Jordan blocks along the diagonal, each block being $J=\mu I + N$ with nilpotent matrix $N$. Explicitly, $$\hat{\Gamma}(a_0)=\begin{bmatrix}
J_1(\mu_1) &         &           &        \\
        & J_2 (\mu_2)&     &        \\
        &         & \ddots    &        \\
        &         &           & J_{l}(\mu_l)
\end{bmatrix}\text{, with each } J_k(\lambda) =
\begin{bmatrix}
\mu_k& 1       &           &        \\
        & \mu_k& \ddots    &        \\
        &         & \ddots    & 1      \\
        &         &           & \mu_k
\end{bmatrix}
= \mu_k I_k + N_k;$$
$\{\mu_1,\mu_2,\dots,\mu_k\}$ are the distinct eigenvalues. Now, as above we can similarly rewrite the joint alignment condition with the transformed questions $\hat X$ and the corresponding parameters, only that the transformed questions $\hat X$, even after the change of basis, cannot be individually incentivized as the $\hat \Gamma(a)$ matrices are not diagonal. Interestingly, they can still be ``locally incentivized'' within \textit{groups} corresponding to the rows of each Jordan block: for each Jordan block $J$ occupying rows (and columns) $K=\{j_1,j_2,\dots,j_k\}$, the corresponding questions $\{j_1,j_2,\dots,j_k\}$ in the transformed question profile can be jointly incentivized without involving/being involved with other questions. Thus, in this case, it is necessary to invoke the generalized BDM mechanism described in \cref{lem: joint sufficiency}.

From a different perspective, if we allow for an action-dependent change of basis, then another series of transformed questions $\tilde{X}(a)=\Gamma(a)^{-1}\left [X(a)-\kappa(a)\right]$ have each of its row individually aligned with $u$, and thus can be individually incentivized. However, notice that the action-dependent change of basis relies on first jointly pinning down the potential $\Gamma(a)$. The above two ways of linear transformation illustrate the improvement from individual alignment, which underpins our discussion on experiment design in \cref{subsec: joint vs individual}.

\section{Joint Alignment: Necessity}\label{sec: necessity}

\subsection{Robust Incentivizability}
We now strengthen incentivizability to a robust version, for which joint alignment remains sufficient. It is shown in the next subsections that this stronger notion also necessitates joint alignment, so that it is fully characterized by joint alignment.

\begin{definition}[Robust incentivizability]\label{def: robust inc MQ}
    An $m$-dimensional question $X$ is robustly jointly incentivizable if there exists invertible matrices $\mu(a)\in\mathbb{R}^{m\times m}$ for each $a\in A$ such that for all $t\in \mathbb{R}^{m\times|\Theta|}$, the question $X_t(a)=X(a)+\mu(a)t$ is incentivizable.
\end{definition}
We can define $m-$robust incentivizability likewise. As an example, when there is only one question, robustness means that incentivizability of $X$ is preserved under proportional perturbations towards any direction $t\in \mathbb{R}^{|\Theta|}$. Intuitively, suppose DM has a constant bias $t(\theta)$ at each state $\theta$. Then, robustness says that the experimenter can find a nearby question $X_t$ respecting this bias, so that DM can still be incentivized to act and report her belief truthfully under this question tailored for him. For multi-dimensional question, robustness means that joint incentivizability of $X$ is preserved when each question $X_j(a)$ at action $a$ is perturbed towards any direction $\mu_j(a) t\in \mathbb{R}^{|\Theta|}$.

Observe that a jointly aligned $X$ is also robustly incentivizable: if $X(a)=\Gamma(a) \left(\lambda{u}(a) +d\right)+\kappa(a)\mathbf{1}$, we can let $\mu(a)=\Gamma(a)$, so that $\tilde{X}(a)=\Gamma(a) \left(\lambda{u}(a) +d+t \right)+\kappa(a)\mathbf{1}$ is still aligned with $u$ and thus incentivizable. 

\begin{proposition}\label{prop: sufficiency MQ}
    If question profile $X$ is jointly aligned with $u$, then it is robustly incentivizable.
\end{proposition}

\subsection{Adjacency, Adjacency Graph, and Kirchhoff's Law}\label{subsec: adjacency}

We now introduce a key notion for incentivizability. Say two actions are \textit{adjacent} if under some belief $p$, both actions are optimal and there is no other optimal action. For intuition, suppose $(a,b)$ are co-optimal for DM under $p$. Then, there is a nearby belief $p'$ such that DM's unique optimal action is either $a$ or $b$ under belief $p'$. For sake of argument, let $a$ be the unique optimal action. Then, facing a belief elicitation question $X$ and payoff $V$, one obvious possible distortion is misreporting $p'$ as $p$ and deviating from $a$ to $b$. Thus, for non-distortionary elicitation of $X$, an important binding constraint is to eliminate such hedging concerning these adjacent pairs. The following lemma, due to PS25, formalizes the above intuition.

Let $\bar{v}=v-\frac{1}{|\Theta|} \sum_{\theta^{\prime} \in \Theta} v\left(\theta^{\prime}\right) \mathbf{1}$ for any $v \in \mathbb{R}^{\Theta}$; geometrically, $\bar v$ is the projection of $v$ onto the $(|\Theta|-1)$-dimensional hyperplane of vectors whose coordinates sum to $0$, which we denote as $\Delta$. Let $\Delta_a^b=\bar{u}(b)-\bar{u}(a)$ be the payoff difference vector. For instance, joint alignment between $X$ and $u$ can now be rewritten as $\bar X(a)=\Gamma(a) \left(\lambda \bar{u}(a) +\bar d\right)$, or equivalently, $\Gamma(b)^{-1}\bar{X}(b)-\Gamma(a)^{-1}\bar{X}(a)=\lambda\Delta_a^b$ for all $a,b\in A$.

\begin{lemma}[Adjacency Lemma, generalizing  Lemma 4 of PS25]\label{lemma: adjacency MQ}
    Suppose that $X$ is jointly incentivizable. If actions a and b are adjacent, then there are matrices $\lambda\in\mathbb{R}^{m\times 1}$ and nonzero $\Gamma\in\mathbb{R}^{m\times m}$ such that
    \begin{equation}
        \bar X(b)=\lambda_{ab}\Delta_a^b+\Gamma_{ab}\bar X(a).\label{eq: adjacency lemma MQ}
    \end{equation}
    \end{lemma}
The lemma thus states a ``local'' joint alignment between $X$ and $u$: $\Gamma_{ab}$ and $\lambda_{ab}$ both depend on the pairing between $(a,b)$, whereas the ``global'' joint alignment condition above requires a constant $\lambda$ and only allows for $\Gamma$ depending on the action itself. For instance, in the action-revelation mechanism, we have $X(a_j)=\sigma_{j-i+1}X(a_i)$, i.e. . To globalize \cref{eq: adjacency lemma MQ}, we construct an \textit{adjacency graph} $G=(V, E)$ with vertices $V=A$ and edges $E=\{(a, b) \in A \times A:\, (a, b) \text{ are adjacent}\}$ and examine $\Gamma_{ab}$ as the length of edges as defined below.

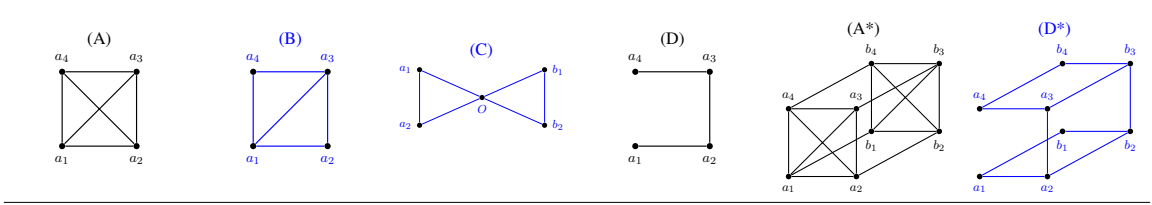
\begin{table}[!htbp]
    \centering
    \caption{Adjacency Graphs}
    \label{tab:AdjGraph}

    \begin{adjustbox}{max totalsize={\textwidth}{\textheight},center}
    \begin{tabular}{C{\cellw} C{\cellw} C{\cellw} C{\cellw} C{\cellw} C{\cellw}}
    \hline
    \\
    \begin{minipage}[c]{\cellw}\centering
      {(A)}\\
    \adjustbox{valign=c,max size={\cellw}{\rowH}}{%
    \begin{tikzpicture}[baseline=(current bounding box.center)]
      \def\s{1.8}
      \node[vertex, label={[vlabel]below:$a_1$}] (v1) at (0,0) {};
      \node[vertex, label={[vlabel]below:$a_2$}] (v2) at (\s,0) {};
      \node[vertex, label={[vlabel]above:$a_3$}] (v3) at (\s,\s) {};
      \node[vertex, label={[vlabel]above:$a_4$}] (v4) at (0,\s) {};
      \draw[edge] (v1)--(v2) (v2)--(v3) (v3)--(v4);
      \draw[edge] (v1)--(v3) (v1)--(v4) (v2)--(v4);
    \end{tikzpicture}}
    \end{minipage}
    &
    \begin{minipage}[c]{\cellw}\centering
    \color{blue}(B)\\
        \adjustbox{valign=c,max size={\cellw}{\rowH}}{%
        \begin{tikzpicture}[baseline=(current bounding box.center)]
      \def\s{1.8}
      \node[vertex, label={[vlabel]below:$a_1$}] (v1) at (0,0) {};
      \node[vertex, label={[vlabel]below:$a_2$}] (v2) at (\s,0) {};
      \node[vertex, label={[vlabel]above:$a_3$}] (v3) at (\s,\s) {};
      \node[vertex, label={[vlabel]above:$a_4$}] (v4) at (0,\s) {};
      \draw[edge] (v1)--(v2) (v2)--(v3) (v3)--(v4);
      \draw[edge] (v1)--(v4) (v1)--(v3);
        \end{tikzpicture}}
    \end{minipage}
    &
    \begin{minipage}[c]{\cellw}\centering
    \color{blue}(C)\\
    \adjustbox{valign=c,max size={\cellw}{\rowH}}{%
        \begin{tikzpicture}[baseline=(current bounding box.center)]
      \node[vertex, label={[vlabel]below:$O$}] (o) at (0,0) {};
      \node[vertex, label={[vlabel]left:$a_1$}]  (a1) at (-1.8, 0.8) {};
      \node[vertex, label={[vlabel]left:$a_2$}]  (a2) at (-1.8,-0.8) {};
      \node[vertex, label={[vlabel]right:$b_1$}] (b1) at ( 1.8, 0.8) {};
      \node[vertex, label={[vlabel]right:$b_2$}] (b2) at ( 1.8,-0.8) {};
      \draw[edge] (a1)--(a2)--(o)--(a1);
      \draw[edge] (b1)--(b2)--(o)--(b1);
        \end{tikzpicture}}
    \end{minipage}
    &
    \begin{minipage}[c]{\cellw}\centering
    {(D)}\\
    \adjustbox{valign=c,max size={\cellw}{\rowH}}{%
        \begin{tikzpicture}[baseline=(current bounding box.center)]
      \def\s{1.8}
      \node[vertex, label={[vlabel]below:$a_1$}] (v1) at (0,0) {};
      \node[vertex, label={[vlabel]below:$a_2$}] (v2) at (\s,0) {};
      \node[vertex, label={[vlabel]above:$a_3$}] (v3) at (\s,\s) {};
      \node[vertex, label={[vlabel]above:$a_4$}] (v4) at (0,\s) {};
      \draw[edge] (v1)--(v2) (v2)--(v3) (v3)--(v4);
        \end{tikzpicture}}
    \end{minipage}
    &
    \begin{minipage}[c]{\cellw}\centering
   {(A*)}\\
    \adjustbox{valign=c,max size={\cellw}{\rowH}}{%
        \begin{tikzpicture}[baseline=(current bounding box.center)]
      \def\s{1.8}
      \node[vertex, label={[vlabel]below:$a_1$}] (a1) at (0,0) {};
      \node[vertex, label={[vlabel]below:$a_2$}] (a2) at (\s,0) {};
      \node[vertex, label={[vlabel]above:$a_3$}] (a3) at (\s,\s) {};
      \node[vertex, label={[vlabel]above:$a_4$}] (a4) at (0,\s) {};
      \draw[edge] (a1)--(a2) (a2)--(a3) (a3)--(a4);
      \draw[edge] (a1)--(a3) (a1)--(a4) (a2)--(a4);
      \begin{scope}[shift={(2.2,1.2)}]
        \node[vertex, label={[vlabel]below:$b_1$}] (b1) at (0,0) {};
        \node[vertex, label={[vlabel]below:$b_2$}] (b2) at (\s,0) {};
        \node[vertex, label={[vlabel]above:$b_3$}] (b3) at (\s,\s) {};
        \node[vertex, label={[vlabel]above:$b_4$}] (b4) at (0,\s) {};
        \draw[edge] (b1)--(b2) (b2)--(b3) (b3)--(b4);
        \draw[edge] (b1)--(b3) (b1)--(b4) (b2)--(b4);
      \end{scope}
      \draw[edge] (a1)--(b1) (a2)--(b2) (a3)--(b3) (a4)--(b4);
        \end{tikzpicture}}
    \end{minipage}
    &

    \begin{minipage}[c]{\cellw}\centering
    \color{blue}(D*)\\
    \adjustbox{valign=c,max size={\cellw}{\rowH}}{%
        \begin{tikzpicture}[baseline=(current bounding box.center)]
      \def\s{1.8}
      \node[vertex, label={[vlabel]below:$a_1$}] (a1) at (0,0) {};
      \node[vertex, label={[vlabel]below:$a_2$}] (a2) at (\s,0) {};
      \node[vertex, label={[vlabel]above:$a_3$}] (a3) at (\s,\s) {};
      \node[vertex, label={[vlabel]above:$a_4$}] (a4) at (0,\s) {};
      \draw[edge] (a1)--(a2) (a2)--(a3) (a3)--(a4);
      \begin{scope}[shift={(2.2,1.2)}]
        \node[vertex, label={[vlabel]below:$b_1$}] (b1) at (0,0) {};
        \node[vertex, label={[vlabel]below:$b_2$}] (b2) at (\s,0) {};
        \node[vertex, label={[vlabel]above:$b_3$}] (b3) at (\s,\s) {};
        \node[vertex, label={[vlabel]above:$b_4$}] (b4) at (0,\s) {};
        \draw[edge] (b1)--(b2) (b2)--(b3) (b3)--(b4);
      \end{scope}
      \draw[edge] (a1)--(b1) (a2)--(b2) (a3)--(b3) (a4)--(b4);
        \end{tikzpicture}} 
    
    \end{minipage}
    \\
    \hline
    \end{tabular}
    \end{adjustbox}

\medskip
{\footnotesize Several representative graphs. (1) A single MCQ task with $\Theta = A$ and payoff $u(a;\theta)=\mathbf{1}\{a=\theta\}$. $G$ is generally complete as in (A) because any two answers can be co‑optimal at some belief. (2) Numerical guess. Take $\Theta$ to be a discrete interval and let $A$ be numeric guesses on the same grid, and the decision payoff be a strictly proper score such as  $u(a;\theta)=-(a-\theta)^2$. The induced $G$ is a tree as in (D) as only neighboring guesses can be co‑optimal. (3)  A series of MCQ. For each MCQ $i$, $\Theta_i = A_i$, payoff $u_i(a;\theta)=\mathbf{1}\{a_i=\theta_i\}$, and $u=\sum_i u_i$. The corresponding adjacency graph $G$ is a product of complete graphs $G_i$ as in (A$^*$). (4) A series numerical guesses. Likewise, the adjacency graph $G$ is a product of trees as in (D$^*$).}
\end{table}
For a subspace $D$ of $\Delta$, let $D^\perp$ be its orthogonal in $\Delta$, and write $\bar{v}_D^\perp$ as the projection of $\bar{v}$ in $D^\perp$ for a vector $v$. For adjacent actions $a$ and $b$, write $\bar{X}(a)^\perp$ shorthand for $\bar{X}(a)^\perp_{\Delta_a^b}$, and $\bar{X}(b)^\perp$ similarly. For an edge $(a,b)\in E$, when $\bar{X}(a)^\perp$ has full rank $m$, define the \textit{edge length matrix} $w(a,b)$ as the unique matrix that satisfies $\bar{X}(b)^\perp=w(a,b)\bar{X}(a)^\perp$. Due to the orthogonal decomposition, $\Gamma_{ab}=w(a,b)$ is also a solution to \cref{eq: adjacency lemma MQ}. For instance, for a single question $X$, when $\bar{X}(a)$ and $\bar{X}(b)$ are not colinear, the edge length $w(a, b)$ equals the ratio between ${p_{ab}\cdot \bar{X}(a)}$ and ${p_{ab}\cdot \bar{X}(b)}$ for a vector $p_{a b}$ orthogonal to the payoff difference $\Delta_a^b$. Intuitively, this is the ratio between the length of the orthogonal projections of $\bar{X}(a)$ and $\bar{X}(b)$ on the hyperplane perpendicular to $\Delta_a^b$. 

When $\bar{X}(a)^\perp$ does not have full rank (e.g., colinearity in the single question case), there remains freedom for defining the edge length $w$. Our goal is a proper calibration of such $w$ through perturbation by virtue of robustness in incentivizability, such that the proposition below holds and leads to global joint alignment. Define the product length for a path $P=(v_0,v_1,\dots,  v_n)$ as $\prod_{i=1}^n w(v_{i-1},v_i)=w(v_{n-1},v_n)w(v_{n-2},v_{n-1})\cdots w(v_0,v_1)$. 

\begin{proposition}
\label{prop: gamma potential MQ}
If for all cycles of $G$, the edge length product equals to identity , there exists a potential $\Gamma(a)\in \mathbb{R}^{m\times m}$ invertible for each $a$ and a column vector $\lambda_{ab}\in \mathbb{R}^{m}$ for each adjacent pair $a$ and $b$, such that 
\begin{equation*}\label{eq: potential MQ}
    \Gamma(b)^{-1}\bar{X}(b)-\Gamma(a)^{-1}\bar{X}(a)=\lambda_{ab}(\bar{u}(b)-\bar{u}(a)) .
\end{equation*}
\end{proposition}
In other words, robust incentivizability requires that the deviation in question profile is in line with the deviation gain in the original payoffs, subject to factorization by a potential $\Gamma$. Again let's take the single-question $X$ as an example. Intuitively, the ratio of $\Gamma(b)$ against $\Gamma(a)$ captures the difference in the respective projection of $\bar{X}(b)$ and $\bar{X}(a)$ in the space orthogonal to $\Delta_a^b$. Then, this factorized difference of $\bar{X}(b)$ and $\bar{X}(a)$ serves as a residual in the span of $\Delta_a^b$. As a consequence of the proposition, for a cycle $C=(a_0,a_1,\cdots,a_{n-1},a_n=a_0)$, 
\begin{equation}\label{eq: zero sum}
    \sum_{i=0}^{n-2}(\lambda_{i,i+i}-\lambda_{n-1,0})\Delta_i^{i+1}=0\text{,}
\end{equation}
\noindent indicating the sum of these residual vectors from different spans of $\Delta$ must be null given that the orthogonal components already cancel out. Then, if the $\Delta$ vectors were ``sufficiently independent'', then the $\lambda$ differences would have to be all zero. In the next subsections, we will exploit this property for different graphs aiming at unifying $\lambda$ across edges $(a,b)$.

To establish the prerequisite in \cref{prop: gamma potential MQ}, we borrow the famous Kirchhoff's voltage law in physics for intuition. Kirchhoff's law states the following: the voltage drop $g$ between nodes in any electric circuit constitutes a scalar field if and only if the sum of voltage drop along every cycle is zero, in which case $g$ is the gradient for an (electric) potential $f$ unique up to an additive constant and one can obtain $f$ by solving $\operatorname{div}(\operatorname{grad} f)=\operatorname{div} g$. The lemma below generalizes the intuition of Kirchhoff's law and is a standard result in combinatorial graph theory (e.g., Lemma 5.3 of \citet{Zaslavsky1989BiasedGraphsI}) and lattice gauge theory (see, e.g., \citet{Rothe2012LGT}). 
\begin{table}[!htbp]
\caption{Three equivalent properties}

\begin{tabular}{@{}>{\centering\arraybackslash}m{0.31\textwidth}@{\hspace{0.035\textwidth}}
                >{\centering\arraybackslash}m{0.31\textwidth}@{\hspace{0.035\textwidth}}
                >{\centering\arraybackslash}m{0.31\textwidth}@{}}

\hline
\\
\begin{minipage}{\linewidth}\centering
\textbf{\footnotesize Zero sum length}\\

\begin{tikzpicture}[scale=1.0]
  \coordinate (A) at (0,0);
  \coordinate (B) at (1.8,0);
  \coordinate (C) at (0.9,1.45);

  \node[vertex, label={[vlabel]below:$A$}] (n1) at (A) {};
  \node[vertex, label={[vlabel]below:$B$}] (n2) at (B) {};
  \node[vertex, label={[vlabel]right:$C$}] (n3) at (C) {};

  \draw[edge,->] (n1) -- node[vlabel, below] {1} (n2);
  \draw[edge,->] (n2) -- node[vlabel, right] {-3} (n3);
  \draw[edge,->] (n3) -- node[vlabel, left]  {2} (n1);
\end{tikzpicture}\\[-0.35em]
{\scriptsize $1+(-3)+2=0$}
\end{minipage}

&
\begin{minipage}{\linewidth}\centering
\textbf{\footnotesize Path independence}\\
\begin{tikzpicture}[scale=1.0]
  \coordinate (A) at (0,0);
  \coordinate (B) at (1.8,0);
  \coordinate (C) at (0.9,1.45);

  \node[vertex, label={[vlabel]below:$A$}] (n1) at (A) {};
  \node[vertex, label={[vlabel]below:$B$}] (n2) at (B) {};
  \node[vertex, label={[vlabel]right:$C$}] (n3) at (C) {};

  \draw[edge, draw=gray] (n1) -- (n2) -- (n3) -- (n1);

  \draw[edge,->] (n1) -- node[vlabel, below] {1} (n2);
  \draw[edge,->] (n2) -- node[vlabel, right] {-3} (n3);
  \draw[edge,->] (n1) -- node[vlabel, left] {-2} (n3);
\end{tikzpicture}\\[-0.35em]
{\scriptsize $-2=1+(-3)$}
\end{minipage}

&
\begin{minipage}{\linewidth}\centering
\textbf{\footnotesize Potential existence}\\
\begin{tikzpicture}[scale=1.0]
  \coordinate (A) at (0,0);
  \coordinate (B) at (1.8,0);
  \coordinate (C) at (0.9,1.45);

  \node[vertex, label={[vlabel]below:$f(A)=0$}] (n1) at (A) {};
  \node[vertex, label={[vlabel]below:$f(B)=1$}] (n2) at (B) {};
  \node[vertex, label={[vlabel]right:$f(C)=-2$}] (n3) at (C) {};

  \draw[edge, draw=gray] (n1) -- (n2) -- (n3) -- (n1);

  \draw[edge] (n1) -- (n2);
  \draw[edge] (n2) -- (n3);
  \draw[edge] (n3) -- (n1);
\end{tikzpicture}\\[-0.35em]
{\scriptsize $g(A\to B)=f(B)-f(A)$}
\end{minipage}

\\
\hline
\end{tabular}
\end{table}

\begin{lemma}\label{lem: Kirchhoff MQ}
    The following two properties are equivalent:
\begin{enumerate}[itemsep=1pt,topsep=3pt,leftmargin=*]
    \item Every cycle in $G$ has length product $I_m$. 
    \item Existence of a matrix potential function: there exists a function $\Gamma:V\to\mathbb{R}^{m\times m}$ such that for every edge $(v,v')\in E$, $w(a,b)=\Gamma(b)\Gamma(a)^{-1}$. Moreover, $\Gamma$ is unique up to a right multiplication of an invertible matrix.
\end{enumerate}
\end{lemma}

The remaining subsections are organized as follows: the next subsection completes the proof of \cref{prop: gamma potential MQ}; \cref{subsec: 1d graph} builds on \cref{prop: gamma potential MQ} and first derives the necessity of joint alignment (\cref{thm: one block alignment MQ}) for the baseline graphs; \cref{subsec: 1b product} generalizes the result  (\cref{thm: product general MQ}) to product graphs; \cref{subsec: multiblock} further studies the more general graphs.

\subsection{Proof of \cref{prop: gamma potential MQ}: Cycle Basis and Potential}

We first introduce the mathematical tools in graph theory for our analysis. In the sequel, we consider finite undirected graphs, i.e. undirected graphs with a finite number of vertices and edges, unless otherwise specified. For a finite undirected graph $G=(V, E)$ with vertices $V$ and edges $E$, we fix an arbitrary orientation for each edge $e\in E$ and define the \textit{vertex-edge incidence matrix}
$$
A=\left[a_{v,e}\right]_{V \times E}, \quad a_{v,e}= \begin{cases}+1 & \text { if } e \text { enters } v, \\ -1 & \text { if } e \text { leaves } v, \\ 0 & \text { otherwise. }\end{cases}
$$
The equation $A x=0$ enforces inflows being equal to outflows at every vertex. Thus, $\operatorname{ker} A$ is the space spanned by all cycles. A basis of $\operatorname{ker} A$ therefore corresponds to a set of cycles that spans all cycles in the graph.

\begin{definition}[Cycle basis, MCB and elementary cycle]
Let $G=(V, E)$ be a graph with incidence matrix $A$. For a cycle $C\subset E$, define its incidence vector $x_C\in\mathbb{R}^{E}$ with $\left(x_C\right)_e=\mathbf{1}\{e \in C\}$. A \textbf{cycle basis} of $G$ is a set of cycles $\mathcal{C}=\left\{C_1, \ldots, C_k\right\}$ such that the vectors $\left\{x_{C_1}, \ldots, x_{C_k}\right\}$ form a linear basis of $\operatorname{ker} A$. 

For a cycle $C$, define its length $|C|$ as the number of arcs in it. Among all cycle bases, a \textbf{minimum cycle basis} (MCB) is one that minimizes the total length $\sum_{C_i\in \mathcal{C}}\left|C_i\right|$. For an MCB $\mathcal{C}$, $C$ is an \textbf{elementary cycle} if $C\in\mathcal{C}$.
\end{definition}

 A vertex $v$ is a cut-vertex if the graph $G-\{v\}$ is disconnected.  A block of $G$ is a maximal connected subgraph with no cut-vertices.\footnote{PS25 use the terms \emph{splitting action} and \emph{splitting collection} for the same graph-theoretic ideas: a splitting action is a cut‑vertex, and a splitting collection is the set nodes of a block.}
 \Cref{tab:AdjGraph} summarizes how edge density relates to block decompositions across standard graph families. At one end, the complete graph is maximally dense and consists of a single block; at the other, a tree is maximally sparse with each edge itself constituting a block. As the number of edges decreases from a complete graph, the graph becomes sparser and the number of blocks typically increases. By standard graph theoretical results (see \Cref{lemma:MCBdecompose} in \cref{subsec: multiblock}), any graph and its MCB can be decomposed into the edge-disjoint union of its blocks, with the graph’s MCB equal to the union of the MCB of those blocks. Hence, in what follows, we focus on the one-block case and later extend the analysis to multiple blocks by invoking \Cref{lemma:MCBdecompose}.

The assumption below is maintained throughout.
\setcounter{assumption}{0}

\begin{assumption}\label{assumption:cycle-length MQ}
There exists an MCB of $G$ such that the length of each of its elementary cycle is lower than $|\Theta|-(m-1)$.
\end{assumption}
For a single-question $X$, when the graph is complete, the assumption requires only $ m<|\Theta|-2$; it also holds for all the examples in PS25, where the action space is the same as the state space and there does not exist a grand elementary cycle involving all actions (states).\footnote{PS25 notice that ``the Adjacency Lemma has no bite when $|\Theta|= 2$, and limited bite when $|\Theta|= 3$''. See their Section 8.1 for a detailed discussion. On the other hand, notice that when $m=|\Theta|-1$ or is larger, then the experimenter can effectively elicit DM's entire belief as long as the questions are not redundant. Hence, the only case left out is when $m=|\Theta|-2$.}
 
\begin{lemma}\label{lemma: product 1}
    Under \cref{assumption:cycle-length MQ}, if $X$ is robustly incentivizable, then edge length can be defined such that the edge length product along all cycles in the MCB above is identity.
\end{lemma}
The proof is deferred to \cref{appendix: proof}
\begin{proof}[Proof of \Cref{prop: gamma potential MQ}]
By definition of MCB, when all cycles in an MCB have identity edge length product, so does any cycle in the graph. By \cref{lem: Kirchhoff MQ}, there exists a matrix potential $\Gamma$ on $G$: fix an action $a_0$ with $\Gamma(a_0)=I_{m}$; for any $a$ and any path $P=(a_0, a_1, \dots, a_{n-1}, a_n=a)$ connecting it with $a_0$, set $\Gamma(a)=\prod_{k=1}^{n}w(a_{k-1},a_k)\Gamma(a_0)$, so that for every adjacent $a$ and $b$, $w(a,b)=\Gamma(b)\Gamma(a)^{-1}$. Rewriting \cref{eq: adjacency lemma MQ}, $\bar{X}(b)=\lambda_{ab} \Delta_a^b+\Gamma(b)\Gamma(a)^{-1}\bar{X}(a)$ $\Rightarrow   \Gamma(b)^{-1}\bar{X}(b)-\Gamma(a)^{-1}\bar{X}(a)=\Gamma(b)^{-1}\lambda_{ab}\Delta_a^b.$
Letting $\lambda_{ab}=\Gamma(b)^{-1}\lambda_{ab}$ yields \cref{eq: potential MQ}.
\end{proof}

\subsection{One-Block Graph}\label{subsec: 1d graph}

We first derive necessary conditions for robust incentivizability within in the minimal unit --- a block. 
The following assumption requires that the payoffs for action within any elementary cycle are differentiated enough in the sense that there is no colinearity across gains from deviation.
\begin{assumption}\label{assumption:1d-independence}
For an MCB of $G$ and any of its elementary cycles $C$, the family of vectors $\left\{\Delta_{a}^{b}:\left(a \rightarrow b\right) \in C\right\}$ has rank $\left|C\right|-1$.
\end{assumption}
Under the assumption, robust incentivizability implies that the question $X$ and the payoff $u$ must be aligned. 

Applying the proposition to one-block graphs, we can unify $\lambda$ across all cycles.

\begin{theorem}\label{thm: one block alignment MQ}
Under \Cref{assumption:1d-independence,assumption:cycle-length MQ}, for a one-block adjacency graph $G$, $X$ is robustly jointly incentivizable if and only if $X$ is jointly aligned with $u$.
\end{theorem}

\begin{proof}
We show necessity as sufficiency is established in \cref{prop: sufficiency MQ}. Consider an MCB of $G$ satisfying \cref{assumption:cycle-length MQ} and an elementary cycle $C$. Notice that the $\Delta_i^{i+1}$ vectors in \cref{eq: zero sum} are linearly independent by \cref{assumption:1d-independence}, so for any adjacent actions $a$ and $b$ in $C$, $\lambda_{ab}=\lambda_{C}$ for some constant $\lambda_C$. Because $G$ is a one-block graph, by definition any two elementary cycles share at least a common edge, so that any two elementary cycles share a common $\lambda$ at this edge.\footnote{See  \Cref{lemma:2vc} in \cref{appendix: maths} for more equivalent conditions for one-block graph.} Hence, the $\lambda$ is uniform within the entire MCB $\mathcal{C}$, so that for any pair of actions $a$ and $b$, there exists a $\lambda$ such that
\begin{equation*}
     \Gamma(b)^{-1}\bar{X}(b)-\Gamma(a)^{-1}\bar{X}(a)=\lambda(\bar{u}(b)-\bar{u}(a)) .
\end{equation*}
Hence, $\Gamma(a)^{-1}\bar{X}(a)-\lambda\bar{u}(a)$ is a constant matrix for all $a$. Similar to the one-dimensional case, it follows that there exist $d\in \mathbb{R}^{m\times|\Theta|}$ such that $\bar{X}(a)=\Gamma(a) \left( \lambda \bar{u}(a) + \bar{d} \right) $, so that $X$ is jointly aligned with $u$.
\end{proof}
The last step above is reminiscent of revenue equivalence in mechanism design where we start with the lowest type whose IR constraint is binding and accumulate the payoff difference along the path to a given type. It is well known that the payment rule corresponds to the node potential of the shortest path polyhedron; see \citet{vohra2011mechanism}.

PS25 show that under their assumptions (quoted below in \cref{cor: PS complete}), for a \textit{complete} adjacency graph with a single question $X$, individual alignment is necessary for single-question incentivizability. \cref{thm: one block alignment MQ} implies that in fact alignment fully characterizes robust incentivizability, and the assumption of completeness is overly demanding. For incomplete graphs, alignment are also necessary for incentivizability, provided that the graph remains one-block, i.e. without a cut-vertex.
\begin{corollary}[PS25 Theorem 2, necessity]\label{cor: PS complete}
    For a single-question $X$, suppose that (i) the adjacency graph is complete and $|A| \ge 4$, and (ii) for any four distinct actions $a, b_0, b_1, b_2$, the set of vectors $\{\Delta^{b_i}_a \}_{i=0,1,2}$ is linearly independent. Then, $X$ is robustly incentivizable if and only if it is aligned with $u$.
\end{corollary}
To see \cref{assumption:cycle-length MQ,assumption:1d-independence} both hold under the ones above, first note that for a complete graph with $|A|\ge 4$, every elementary cycle in an MCB is a triangle;\footnote{See \Cref{lemma:completeMCB} in \cref{appendix: maths}.} thus, condition (ii) below implies that (1) the number of states must be at least $4$ which is larger than the length of elementary cycles, 3, giving \cref{assumption:cycle-length MQ}, and (2) for each elementary cycle $C=(a,b, c,a)$ the family $\{\Delta_a^{b},\Delta_b^{c},\Delta_c^{a}\}$ has rank $2=|C|-1$, giving \cref{assumption:1d-independence}. 

\medskip\noindent
\textbf{Remark.} We now review the idea of $m$-(robust) incentivizability in \cref{subsec: joint vs individual}. Consider a state-contingent task $(\Theta,A,u)$ with adjacency graph satisfying the two assumptions. \cref{thm: one block alignment MQ} now implies that for a question $Y$ to be $m$-(robust) incentivizable for some $m$ small enough for \cref{assumption:cycle-length MQ}, it is in fact both necessary and sufficient that $Y$ is weakly aligned with $u$. Hence, although \cref{prop: m-incen} and \cref{cor: m-decomposable} have provided useful recipes, if the experimenter wishes to incentivize truthful report of $r(a)=\mathbb{E}_pY(a,\theta)$ with a small $m$, the premise must hold first that $Y$ is weakly aligned $u$.

\section{More General Cases}\label{sec: general graph}

\subsection{One-Block Product Graph}\label{subsec: 1b product}

We now consider a special type of one-block graph beyond \cref{assumption:1d-independence}. A \textit{product problem} consists of multiple tasks $i\in I$, with the corresponding state and action spaces $\Theta=\times_i \Theta_i, A=\times_i A_i$ where $|\Theta_i|\geq2$. The payoff $u(a ; \theta)=\sum_i u_i\left(a_i, \theta_i\right)$ for $i\in\{1,2,\cdots, I\}$. Given independence between different tasks, two actions $a$ and $b$ are adjacent only if they differ in at most one task. The adjacency graph $G$ can thus be written as the Cartesian product of task adjacency graphs $G_i$. Notice that an MCB of $G$ consists of the following two categories of cycles: (1) 4-cycles $((a^i,a^j,a^{-ij}),(b^i,a^j,a^{-ij}),(b^i,b^j,a^{-ij}),(a^i,b^j,a^{-ij}))$ where neighboring actions differ alternatively in tasks $i$ and $j$, and (2) cycles $(a_0,a_1,a_2,\dots,a_{n-1},a_n=a_0)$ that only differ within one task $i$, which are themselves elementary cycles for task $i$.\footnote{For the special case of each $G_i$ being complete graphs, and a more general result, see \Cref{lemma:completeMCB,lemma:productMCB} in \cref{appendix: maths}.}

For product graphs, the first category of 4-cycles mechanically fails \cref{assumption:1d-independence} as the four payoff difference vectors can take only two values, $\Delta_{a^i}^{b^i}$ and $\Delta_{a^j}^{b^j}$. However, we can still obtain a more admissive, taskwise necessary alignment condition under the following weakening.
\begin{assumption}\label{Assumption: md independence}
For $G=\bigtimes_iG_i$, let $\mathcal{B}^4$ be the set of 4-cycle MCB defined above in the first category, and $\mathcal{B}$ be an MCB of $G_i$. Assume that (1) for each $C\in\mathcal{B}^4$, the set $\left\{\Delta_{a}^{b}:\left(a,b\right) \text{ are neighboring nodes in } C\right\}$ has rank 2, and (2) for each $i$ and $C\in\mathcal{B}_i$, the set $\left\{\Delta_{a_i}^{b_i}:\left(a_i, b_i\right) \text{ are neighboring nodes in } C\right\}$ has rank $\left|C\right|-1$.\footnote{Here is a unified perspective for two assumptions: assume that for any elementary cycle $C$ of $G$, define $S_C=\left\{\Delta_{a}^{b}:\left(a \rightarrow b\right) \in C\right\}$, $\operatorname{dim} \operatorname{ker} S_C = n_C$. Here $n_C$ refers to the number of tasks that cycle $C$ crosses and is also the number of "trivial dependence". Specifically, in assumption 1, all cycles exist in only one task, the sum of all $\Delta$ is 0 yielding one trivial dependence. So $\operatorname{dim} \operatorname{ker} S_C = 1 \Leftrightarrow \operatorname{rank} S_C=|C|-1$; in assumption 2, for those 4-cycles that span two tasks, $\Delta_a^b+\Delta_{c}^{d}=0$ and $\Delta_b^c+\Delta_d^a=0$. So $\operatorname{dim} \operatorname{ker} S_C = 2 \Leftrightarrow \operatorname{rank} S_C=2$. Intuitively, all elementary cycles are required to have their respective maximally possible rank.}
\end{assumption}

\begin{definition}
   $X$ is taskwise jointly aligned with $u$ if there exists a matrix $d\in \mathbb{R}^{m\times|\Theta|}$, a column vector $\lambda_i\in \mathbb{R}^m$ for each task $i$, and for each $a\in A$, an invertible matrix $\Gamma(a) \in \mathbb{R}^{m \times m} $ and a column vector $\kappa(a) \in \mathbb{R}^m$, such that $X(a)=\Gamma(a)\big(\sum_i\lambda_i u_i(a)+d\big)+\kappa(a)\mathbf{1}$ for all $a\in A$.
\end{definition}
\begin{theorem}\label{thm: product general MQ}
Under \Cref{assumption:cycle-length MQ,Assumption: md independence}, for a product graph $G=\bigtimes_iG_i$ with each $G_i$ being one-block, $X$ is robustly incentivizable if and only if $X$ is taskwise jointly aligned with $u$.
\end{theorem}
The proof is deferred to \cref{appendix: proof}. Theorem 3 of PS25 studies conditions for single question incentivizability when each $G_i$ is complete, under assumptions that imply \cref{assumption:cycle-length MQ,Assumption: md independence}.\footnote{\cref{assumption:cycle-length MQ} holds as $|\Theta|\ge 8$ (at least 3 tasks with at least 2 states within each task), while an elementary cycle has length 3 or 4 in such a graph. Part (1) in \cref{Assumption: md independence} holds through the same argument in their footnote 11 and the beginning of their Appendix E.1, while part (2) is equivalent to their part (i) for complete $G_i$ as a within-task elementary cycle has length three. (Adjacency in the binary action case is trivially true.)} As with \cref{thm: one block alignment MQ}, we can strengthen incentivizability to its robust version as a natural corollary of the above \cref{thm: product general MQ}. Interestingly, they require a minimal number of 3 tasks for carrying out their proofs, while admitting that they ``\textit{do not know whether the conclusion of Theorem 3 holds}'' for the case of two tasks. Noticing that the case of one task is no other than \cref{thm: one block alignment MQ} and its corollary, our approach fills this gap by clarifying the role played by two types of elementary cycles. 

\subsection{Multiple-Block Graph}\label{subsec: multiblock}
In the previous sections, we have finished the analysis of the single block graph. As we delete even more edges from a graph, cut-vertices emerge and the graph eventually breaks into multiple blocks. The lemma below (with proof in \cref{appendix: maths}) provide a way of partitioning an MCB into different blocks.

\begin{lemma}[MCB blockwise decomposition]\label{lemma:MCBdecompose}
Let $G$ be a connected graph, with vertex-edge incidence matrix $A$ and its block partition $E=\bigsqcup_{i=1}^r B_z$. Let $A_z$ and $\mathcal{C}_z$ be the vertex-edge incidence matrix and an MCB of subgraph $B_z$ for $z=1,2,\dots,r$. 
. Then,
\begin{enumerate}[itemsep=1pt,topsep=3pt,leftmargin=*]
    \item  $\operatorname{ker} A=\bigoplus_{z=1}^r \operatorname{ker} A_z$.
    \item $\mathcal{C}=\bigcup_{z=1}^r \mathcal{C}_z$ is an MCB of $G$.
\end{enumerate}
\end{lemma}
By \Cref{lemma:MCBdecompose}, an MCB $\mathcal{C}$ of $G$ splits as a disjoint union of the MCB of the blocks. Therefore, the consistency of $\lambda$ coefficient in the likes of \cref{prop: gamma potential MQ} within a block does not pass onto another through the cut vertices. Although we can no longer reduce all $\lambda$'s to a single value globally, we can still calibrate them blockwise: there exists a blockwise common $\lambda$ within each block. This directly leads to the blockwise alignment condition in \Cref{cor: alignment2} below, which is necessary and sufficient for incentivizability. 

\begin{definition}
    $X$ is blockwise jointly aligned with $u$ if for each block $B_z$ there exists a matrix $d_z\in \mathbb{R}^{m\times|\Theta|}$ and a column vector $\lambda_z\in \mathbb{R}^m$, and for each action $a\in A$, an invertible matrix $\Gamma(a) \in \mathbb{R}^{m \times m} $ and a column vector $\kappa(a) \in \mathbb{R}^m$, such that ${X}(a)=\Gamma(a) \left( \lambda_z{u}(a) + {d}_z \right)+\kappa(a)\mathbf{1} $ when $a$ is in block $B_z$.
\end{definition}

\begin{proposition}\label{cor: alignment2}
Under \Cref{assumption:1d-independence}, for an adjacency graph $G$ with block partition $G=\bigcup_{z=1}^Z B_z$, $X$ is robustly incentivizable if and only if it is blockwise jointly aligned with $u$.
\end{proposition}
The necessity follows from a straightforward restriction of the arguments earlier in \cref{subsec: 1d graph} into each block. For sufficiency, we can follow the proof of PS25 by first constructing BDM mechanisms within each block, and then sequentially calibrating the value of BDM payoffs $V_z$ from each block $z$ at each common cut-vertex to obtain a uniform $V$ on $G$.

Similarly, we can extend \cref{thm: product general MQ} to product graphs with multiple blocks in each task.
\begin{definition}
    $X$ is task-block-wise jointly aligned with $u$ if for each task $i$ and block $z_i\in \{1, 2,\cdots,Z_i\}$, there exists $\tau_i\in \mathbb{R}$,$\lambda_{z_i}\in \mathbb{R}^{m\times\Theta}$ and $d_{z_i}\in \mathbb{R}^\Theta$ and for each $a\in A$, an invertible matrix $\Gamma(a) \in \mathbb{R}^{m \times m} $  and a column vector $\kappa(a) \in \mathbb{R}^m$, such that $\bar{X}(a)=\Gamma(a) \left( \sum_i\tau_i\lambda_{z_i}\bar{u}_i(a) + \bar{d}_{z_i} \right) $ for all $a\in A$.
\end{definition}
\begin{proposition}\label{cor:task-block-wise aligned}
    Under \Cref{Assumption: md independence}, for a product graph $G=\bigtimes_iG_i$ with block partitions $G_i=\bigcup_{{z_i}=1}^{Z_i} B_{z_i}$, $X$ is robustly incentivizable if and only if it is task-block-wise jointly aligned with $u$.
\end{proposition}
\begin{example}
    The following is adapted from Examples 2 and 3 in PS25. Consider a variation of MCQ task with actions $A=\left\{a_1, a_2, o, b_1, b_2\right\}$ and states $\Theta=\left\{a_1, a_2, b_1, b_2\right\}$. The payoff $u(a,\theta)$ in  \Cref{tab:payoff-matrix} below represents a coarsened ``matching-the-state'' with a safe option $o$.
\begin{table}[h]
\centering
\caption{Payoff matrix for $u(a,\theta)$}
\label{tab:payoff-matrix}
\begin{tabular}{lcccc}
\hline
\textbf{action \textbackslash{} state} & $a_1$ & $a_2$ & $b_1$ & $b_2$ \\
\hline
$a_1$ & 1   & 0.6 & 0   & 0   \\
$a_2$ & 0.6 & 1   & 0   & 0   \\
$o$   & 0.6 & 0.6 & 0.6 & 0.6 \\
$b_1$ & 0   & 0   & 1   & 0.6 \\
$b_2$ & 0   & 0   & 0.6 & 1   \\
\hline
\end{tabular}
\end{table}

Notice that two actions $a_i$ and $b_j$ can never be co-optimal as any such belief would make the safe option $o$ the unique maximizer. Thus, the adjacency graph is exactly (C) in \Cref{tab:AdjGraph} with two triangle blocks $\left(a_1, a_2, o\right)$ and  $\left(b_1, b_2, o\right)$ sharing the safe action $o$. Though the adjacency graph falls out of the scope of the PS25's assumption of cycle richness, our result can still apply. 
\end{example}
\subsection{Comparison with PS25}\label{sec: PS25 comparison}

On the necessity side, the main technical difference between our method and that of PS25 lies in how to deal with cycles with degeneracy, namely that some of its neighboring nodes $(a_{i-1},a_i)$ having colinear $\big(\bar{X}(a_{i-1}),\bar{X}(a_i)\big)$. PS25 address the issue of degeneracy by building linear independence assumptions on the richness of cycles, which in turn rule out corner cases such as a product graph with three actions in each task and/or exactly two tasks.We instead remove colinearity through perturbing $\bar X$ globally to a non-colinear one (\cref{lemma: product 1}).\footnote{One observation is that in most adjacency graphs including the graphs in Theorems 2-3 of PS25, elementary cycles do not have a full length $|\Theta|$ (\cref{assumption:cycle-length MQ}), which renders the perturbations feasible.} This approach yields a unified treatment of all one-block graphs including both Theorems 2-3 of PS25 via varying the MCB, while replaces the seemingly disconnected assumptions there with a unified \cref{assumption:cycle-length MQ} on the maximal length of elementary cycles. Moreover, this approach requires weaker linear independence assumptions (which accommodate the corner cases of PS25 above) and lends itself to an immediate generalization to multiple questions and multiple blocks.
\section{Related Experimental Literature}\label{sec: discussion}
\begin{table}[!htbp]
\caption{Literature}
\begin{tabular}{c|c|c}
  \FCx{}        & \FCc{{\textbf{Single belief elicitation question }}}   & \FCc{{\textbf{Multiple belief elicitation question}}}   \\ \hline
  \FCr{{\textbf{Complete graph}}}   & \FC{\citet{buser2014gender,burks2013overconfidence,niederle2007women,hu2024confidence}} & \FC{\color{blue} \citet{exley2022gender,exley2024gender,coffman2024stereotypes,mobius2022managing,costa2008stated}} \\ \hline
  \FCr{{\textbf{Incomplete graph}}}   & \FC{ \color{blue} \citet{enke2023cognitive,baldiga2014gender}} & \FC{\color{blue} \citet{enke2023confidence,nunnari2025audi}} \\ 
\end{tabular}
\end{table}
Beyond the single-question studied by PS25, prominent experimental designs in prior work repeatedly fall into three classes. First, many MCQ-based studies elicit multiple belief statistics about performance—such as the expected score, categorical self-assessments, or relative rank—in order to study confidence and self-promotion (\citet{exley2022gender,exley2024gender,coffman2024stereotypes}) and to track confidence dynamics over time and feedback (\citet{mobius2022managing}) or in strategic settings (\citet{costa2008stated}). Second, a growing set of experiments introduces safe actions like abstain, opt-out and/or coarser choices. This induces an incomplete action structure and can make hedging incentives salient even with a single elicited belief (\citet{enke2023cognitive,baldiga2014gender}). Third, some influential designs combine both features by pairing safe-option environments with multiple belief questions—for example to study selection and aggregate bias (\citet{enke2023confidence}) or selective exposure to information (\citet{nunnari2025audi}). These widely used settings motivate a nondistortionary elicitation characterization that accommodates multiple questions and general adjacency structures. We discuss in detail several typical examples in the following.

\medskip
\textbf{Multiple belief elicitation questions.} In \citet{serra2021mistakes}, participants watch videos and decide for each whether the speaker is telling the truth or lying. One randomly chosen video is paid for correctness; bonus payments also depend on the accuracy of the belief. For each video $i$, the true state and actions $\theta_i,\;a_i \in\{ truth,\; lie \}$. There are two aggregate belief questions: (1) after every group of five videos, ``How many of the five guesses you just made are correct?'' After all 20 videos, ``Compared with previous participants, how well do you think you did?''

In \citet{exley2022gender,exley2024gender}, participant take a multi‑task ASVAB test, where each task is an MCQ. After the tasks, there are some aggregate belief questions including (1) ``How many questions do you think you answered correctly?'', (2) ``Whether your performance would be classified as ``poor'' on the test?'', and (3) ``What is the percent chance of scoring at least 7 out of 10?''

\medskip
\noindent\textbf{Multiple tasks of trees.} \citet{enke2023confidence} offers multi‑dimensional numeric tasks, such as correlation neglect (CN), balls and urns Bayesian updating (BU probability), regression to the mean (RM), and exponential growth bias (EGB), that fit in our tree framework. Each task asks about a point estimate and subjects are paid with standard strictly proper scoring rules. Formally, for the reported probability estimation $a $ about the posterior mean of a random variable $\Theta$, payoff $u$ takes form of the quadratic score: $u(a, \theta)=-(\theta-a)^2$. With belief $p$, the expected payoff $\mathbb{E}_p[u(a, \theta)]=-\operatorname{Var}_p(\theta)-\left(\mathbb{E}_p[\theta]-a\right)^2$ is strictly concave in $a$. When the subjects are asked to choose from a discrete grid of available estimates $A= \left\{0\leq a_1<\cdots<a_m\leq 1 \right\}$, indifference between $a_k$ and $a_{k+1}$ occurs at a single belief threshold, and ties can only occur between neighbors. Hence the adjacency graph is a line graph.

Although the papers didn’t ask about the belief elicitation problem, one might tend to add simple aggregate forecasting questions such as “What do you think your total correct answers are?” or “What do you think your total income is?”. The augmented design falls under our product of multiple block graphs framework and thus \cref{cor:task-block-wise aligned} applies.

\bibliographystyle{ecta-fullname}
\bibliography{refs}

\appendix
\section{Proofs}\label{appendix: proof}

\begin{lemma}\label{lem: matrix decomposition}
   Let $Y$ be an $|A|\times |\Theta|$ matrix. There exist $G\in \mathbb{R}^{|A|\times m}$ and $ D\in \mathbb{R}^{m\times |\Theta|}$ such that $ Y=GD
$ if and only if $\operatorname{rank}(Y)\le m$. 
\end{lemma}
\begin{proof}[Proof of \cref{lem: matrix decomposition}]
If $Y=GD$, then $\operatorname{rank}(Y)=\operatorname{rank}(GD)\le \min\{\operatorname{rank}(G),\operatorname{rank}(D)\}\le m$ so $\operatorname{rank}(Y)\le m$.

Now we prove sufficiency. Let $r=\operatorname{rank}(Y)$. Choose $r$ linearly independent rows of $Y$ and stack them as $D_0\in \mathbb{R}^{r\times |\Theta|}$. Then every row $y_i$ of $Y$ can be written as $y_i=g_i D_0$ for some row vector $g_i\in\mathbb{R}^r$. Let $G_0\in \mathbb{R}^{|A|\times r}$ have rows $g_i$. Then $Y = G_0 D_0$. If $m>r$, it suffices to set
\[
G=\begin{pmatrix} G_0 & 0 \end{pmatrix}\in \mathbb{R}^{|A|\times m},
\qquad
D=\begin{pmatrix} D_0 \\ 0 \end{pmatrix}\in \mathbb{R}^{m\times |\Theta|}.
\]
\end{proof}

\begin{proof}[Proof of \cref{lemma: product 1}]
For a cycle $C=(a_0, a_1, \dots, a_{n-1}, a_n=a_0)$, we write for shorthand $\Delta_{{j-1}}^{j}:=\Delta_{a_{j-1}}^{a_j}$ and define $\Delta_C=\operatorname{span} \{\Delta_{k}^{k+1}\}_{k=0}^{n-1}$. 

Pick an MCB $\mathcal{C}$ satisfying the condition in \cref{assumption:cycle-length MQ}, so that for each $C\in\mathcal{C}$, $\Delta_C$ has dimension no more than $|\Theta|-m-1$. Then, since the number of elementary cycles is finite, $\Delta_{\mathcal{C}}:=\bigcup_{C\in \mathcal{C}}\Delta_C$ does not cover the entire $\Delta$; in fact, we can find $m$ linearly independent vectors $t_1,t_2,\dots, t_m\in\mathbb{R}^{|\Theta|}$ such that each $\bar{t}_j \notin\Delta_\mathcal{C}$ for $j=1,2,\dots, m$; let $t$ be the matrix with $t_j$ as its $j$-th row. Then, $\mu(a)t$ is also full rank by invertibility of $\mu(a)$. By robustness, for any $\varepsilon>0$, the question $X_\varepsilon(a)=X(a)+\mu(a)\varepsilon t$ is incentivizable. Moreover, we can find a sequence $(\varepsilon_k)_{k=1}^\infty\rightarrow 0$ such that for all $a\in A$ and $C\in\mathcal{C}$ containing $a$, the projection of the rows of $\bar{X}_{\varepsilon_k}(a)$ to $\Delta^\perp_C$ have full rank.\footnote{For each $a$ in an elementary cycle $C$, project $\bar{X}_\varepsilon(a)$ to a fixed $m$-dimensional subspace of $\Delta^\perp_C$. Then, the number of $\varepsilon$ such that $\bar{X}_\varepsilon(a)=\bar X(a)+\varepsilon\mu(a) \bar t$ has zero determinant, i.e. not being full rank, is finite. By finiteness of action $a$ and cycle $C$, the set of such undesired $\varepsilon$ is also finite.}

For any elementary cycle $C=(a_0, a_1, \dots, a_{n-1}, a_n=a_0)$, each $\bar{X}_\varepsilon(a_i)^\perp$ in the neighboring pair of actions $(a_i,a_{i+1})$ has full rank for each $\varepsilon$ in the sequence; hence, the edge lengths are already defined by \Cref{lemma: adjacency MQ}. We 
apply \cref{eq: adjacency lemma MQ} iteratively to $(a,b)=(a_{j+1},a_j)$ for $j=0,1,\dots,n-1$ and obtain
 $$\bar{X}_\varepsilon(a_0)=\sum_{j=1}^{n}\prod_{k=j+1}^{n}w(a_{k-1},a_k;\varepsilon)\lambda_{j-1,j}\Delta_{{j-1}}^{j}+\prod_{k=1}^{n}w(a_{k-1},a_k;\varepsilon)\bar{X}_\varepsilon(a_0).$$

\noindent Projecting both sides to $\Delta^\perp_C$, it has to be that $\prod_{k=1}^{n}w(a_{k-1},a_k;\varepsilon)$ is identity as the projection of $\bar{X}_{\varepsilon}(a)$ have full rank.

Finally, for any adjacent actions $a$ and $b$ such that $\bar{X}(a)$ and $\bar{X}(b)$ are colinear, let $w(a,b):=\lim_{m\rightarrow \infty}{w(a,b;\varepsilon_m)}$. As for every cycle $C=(a_0, a_1, \dots, a_{n-1}, a_n=a_0)$, $\bar{X}_{\varepsilon_k}$ satisfies the desired identity product property for all $\varepsilon_k$, by taking $\varepsilon_k\rightarrow 0$, it follows that the edge length product for $\bar{X}$ along cycle $C$ is also identity.
\end{proof}

\begin{proof}[Proof of \cref{thm: product general MQ}]
The sufficiency can be established through a simple taskwise generalization of the BDM mechanism. For necessity, following \Cref{eq: potential MQ}, for a pair of actions $a$ and $a'$ only differing in task $i$, we have
$$
\Gamma(b)^{-1}\bar{X}(b)-\Gamma(a)^{-1}\bar{X}(a)= \lambda_i(a^{-i})\Delta_{a^i}^{b^i}.
$$
Let $(a_1,a_2,a_3,a_4)$ be an elementary cycle with $a_1=(a^i,a^j,a^{-ij}),a_2=(b^i,a^j,a^{-ij}),a_3=(b^i,b^j,a^{-ij}),a_4=(a^i,b^j,a^{-ij})$ that crosses two tasks $i$ and $j$. Summing the above equation for adjacent edges gives
\begin{align*}
0&=\lambda_i(a^j,a^{-ij})\Delta_{a^i}^{b^i}+
\lambda_j(a^i,a^{-ij})\Delta_{a^j}^{b^j}+
\lambda_i(b^j,a^{-ij})\Delta_{b^i}^{a^i} +
\lambda_j(a^i,a^{-ij})\Delta_{b^j}^{a^j}\\
&=[\lambda_i(a^j,a^{-ij})-\lambda_i(b^j,a^{-ij})]\Delta_{a^i}^{b^i}+[\lambda_j(a^i,a^{-ij})-\lambda_j(a^i,a^{-ij})]\Delta_{a^j}^{b^j}.
\end{align*}
By \Cref{Assumption: md independence}, $\Delta_{a^i}^{b^i}$ and $\Delta_{a^j}^{b^j}$ are independent, so that $\lambda_i(a^j,a^{-ij})=\lambda_i(b^j,a^{-ij})$. Fix $j \neq i$ and fix all coordinates $a^{-ij}$. Take any two, possibly non-adjacent values $u^j, v^j \in A^j$. Because the task-$j$ adjacency graph $G_j$ is connected, there is a path $u^j=c_0^j - c_1^j - \cdots - c_m^j=v^j$ in $G_j$. Apply above 4-cycle argument to each adjacent pair $c_t^j - c_{t+1}^j$ and by transitivity of equivalence along the path, $
\lambda_i\left(u^j, a^{-i j}\right)=\lambda_i\left(v^j, a^{-i j}\right)$. Hence $\lambda_i$ does not depend on the value of $a^j$ once fix $a^{-ij}$. 

Now let $s^{-i}, t^{-i} \in \prod_{\ell \neq i} A^{\ell}$ be arbitrary two profiles of the non-$i$ coordinates. Since each $G_{\ell}$ is connected, the product graph $G_{-i}=\bigtimes_{\ell \neq i} G_{\ell}$ is connected, and there exists a path $s^{-i}=z^0, z^1, \ldots, z^r=t^{-i}$ in $G_{-i}$ that changes at most one coordinate at each step. At step $k$, suppose the changing coordinate is $j_k \neq i$, we have $\lambda_i\left(z^k\right)=\lambda_i\left(z^{k+1}\right)$. The transitivity of these equalities along the path gives $\lambda_i\left(s^{-i}\right)=\lambda_i\left(t^{-i}\right)$. Because $s^{-i}, t^{-i}$ were arbitrary, $\lambda_i\left(a^{-i}\right)$ is constant over the entire $\prod_{\ell \neq i} A^{\ell}$. So for each task $i$, there is a unified $\lambda_i$. So for any action $a$ and $b$ that differ in any possible tasks, we have
$$
\Gamma(a)^{-1}\bar{X}(a)-\sum_i \lambda_i u_i\left(a_i\right)=\Gamma(b)^{-1}\bar{X}(b)-\sum_i \lambda_i u_i\left(b_i\right).
$$ Therefore, $X$ is taskwise jointly aligned with $u$.
\end{proof}

\section{Supplemental Results in Graph Theory}
\label{appendix: maths}

\begin{lemma}[Complete graph MCB]\label{lemma:completeMCB}
    Let $K_n$ be the complete graph on vertex set $V=\{1, \ldots, n\}$ 
    . Let $A \in \mathbb{R}^{n \times\frac{ n(n-1)}{2}}$ be its vertex-edge incidence matrix. Then:
\begin{enumerate}
    \item $\operatorname{rank}(A)=n-1$ and $\operatorname{dim} \operatorname{ker} A=\binom{n-1}{2}$.
    \item A minimum-length cycle basis of $\operatorname{ker} A$ consists of the $\binom{n-1}{2}$ 3-cycles through any fixed root $r$, namely the cycles $(r ,i, j)$ for $i,j$ in $V \setminus\{r\}$.
    \end{enumerate}
\end{lemma}



The next result lays out the structure of an MCB for a product graph. The proof follows from Theorem 10 of \citet{imrich2002minimum} and Theorem 14 of \citet{gleiss2003circuit}.
\begin{lemma}[Product graph MCB]\label{lemma:productMCB}
Let $G_1 \times G_2$ be a Cartesian product graph. For each $i\in\{1,2\}$ let $\mathcal{C}_i$ be an MCB of $G_i$, and let $T_i\subseteq \mathcal{C}_i$ be the set of $3$-cycles in $\mathcal{C}_i$. Fix vertices $u_i\in V(G_i)$ and write $\bar i:=3-i$.  

For any cycle $C$ in $G_i$ and any $y\in V(G_{\bar i})$, denote by $C\times_i y$ the copy of $C$ in the $G_i$-fiber over $y$
(i.e.\ $C\times_1 y:=C\times y$ and $C\times_2 y:=y\times C$).
Then a minimum cycle basis $\mathcal{C}^*$ of $G_1 \times G_2$ can be chosen as the union of:
\begin{enumerate}[itemsep=1pt,topsep=3pt,leftmargin=*]
    \item $ \bigcup_{i=1}^2 \left\{\,T\times_i y \ \middle|\ T\in T_i,\ y\in V(G_{\bar i})\,\right\}$;
    \item $ \bigcup_{i=1}^2 \left\{\,C\times_i u_{\bar i} \ \middle|\ C\in \mathcal{C}_i\setminus T_i\,\right\}$;
    \item a suitable independent subset $Q$ of the cross squares $e\times f$ (with $e\in E(G_1)$ and $f\in E(G_2)$).
\end{enumerate}
Write $\lambda(G)$ for the length of the longest cycle appearing in any MCB of $G$. Then
\[
\lambda(G_1 \times G_2)=\max\{4,\lambda(G_1),\lambda(G_2)\}.
\]
\end{lemma}


\Cref{lemma:productMCB} tells us the MCB structure of the Cartesian product graph. First, it keeps every tiny 3-cycles again shortest. Secondly, For longer cycles that live entirely inside $G_1$ or $G_2$, keep just one representative. Finally, add enough cross “squares” to stitch the whole thing together. Every pair of edge in $G_1$ and edge in $G_2$ forms a 4-cycle square. They are the cheapest cross-dimensional cycles which can connect the $G_1$-side to the $G_2$-side cycles without introducing unnecessarily long cycles.


Lemma \ref{lemma:2vc} below connects the notion of cut-vertex with blocks and discusses the structure of a single block. Build the (undirected) \textit{cycle-intersection graph} $H$ as follows: its vertices are the cycles in $\mathcal{C}$; two vertices of $H$ are linked if the corresponding cycles in $\mathcal{C}$ share at least one common edge of $G$. Intuitively, these two cycles in $G$ ``touch'' each other at certain edges.
\begin{lemma}\label{lemma:2vc}
    For a connected graph $G$, the following are equivalent: 
    \begin{enumerate}[itemsep=1pt,topsep=3pt,leftmargin=*]
        \item $G$ has no cut-vertex. 
        \item For every distinct $x, y \in V(G)$, there exist two paths linking $x$ and $y$ with no overlapping vertex.
        \item $G$ has exactly one block.
        \item The cycle-intersection graph $H$ is connected.
    \end{enumerate}
\end{lemma}

\begin{proof}[Proof of \Cref{lemma:2vc}]
(1) $\Longleftrightarrow$ (2) is by \citet[Thm.~3.4.1]{godsil2013algebraic}. (1) $\Longleftrightarrow$ (3) is by definition.

(1) $\Longrightarrow$ (4).
Fix two cycles $C, C^{\prime} \in \mathcal{B}$. Take edges $e \in C$ and $f \in C^{\prime}$. In a graph without cut-vertices, any two edges lie on a common cycle.\footnote{this "fact" is covered in Exercise 3 at §14.16 of \citet{godsil2013algebraic}} Let $D$ be a cycle containing $e$ and $f$. Expressing $D$ as a $\mathrm{GF}(2)$-sum of cycles from $\mathcal{B}$, the resulting symmetric difference is a single cycle; hence the cycles of $\mathcal{B}$ occurring in that sum can be ordered so that consecutive ones share at least one edge. Therefore there is a path in $H$ from the vertex $C$ to the vertex $C^{\prime}$. As $C, C^{\prime}$ were arbitrary, $H$ is connected.

(4) $\Longrightarrow$ (3).
Distinct blocks share no edges, therefore cycles from different blocks share no edges. Consequently, the vertices of $H$ split according to the blocks that contain cycles; if $G$ had two or more such blocks, $H$ would be disconnected-contrary to (4). Thus all cycles of $G$ lie in a single block, and $G$ has exactly one block.
\end{proof}

\begin{proof}[Proof of \Cref{lemma:MCBdecompose}]

(1). Take any $x \in \operatorname{ker} A$. By Cor. 14.2.3, write $x$ as a linear combination of cycle vectors; grouping the cycles by the block in which they lie yields 
$x=\sum_{i=1}^r x_i$, $ x_i \in \operatorname{ker} A_i$, where we regard $x_i$ as a vector on $E(G)$ extended by 0 off $E\left(B_i\right)$. Hence $\operatorname{ker} A \subseteq \sum_i \operatorname{ker} A_i$. The reverse containment is immediate because each $x_i \in \operatorname{ker} A_i$ still satisfies the vertex-balance equations in $G$, so $x_i \in \operatorname{ker} A$.

By §15.4 of \citet{godsil2013algebraic}, every cycle of $G$ lies in a unique block $B_i$. To show the sum is direct, the edge-sets of the blocks are disjoint, so if $\sum_i x_i=0$ with $x_i \in \operatorname{ker} A_i$, then looking at coordinates on $E\left(B_j\right)$ forces $x_j=0$ for every $j$. Thus
$
\operatorname{ker} A=\bigoplus_{i=1}^r \operatorname{ker} A_i
$

(2). By (1), the cycle space decomposes as a direct sum of the cycle spaces of the blocks. Equivalently, every cycle basis $\mathcal{C}$ of $G$ splits uniquely as a disjoint union $\mathcal{C}=\bigsqcup_i \mathcal{C}_i$ where $\mathcal{C}_i$ consists of the cycles of $\mathcal{C}$ contained in $B_i$, and $\mathcal{C}_i$ is a cycle basis of $B_i$.

The length of any cycle basis is the sum of the lengths of its blockwise parts, and cycles cannot mix edges from different blocks. Therefore minimizing the total length over all bases of $G$ is equivalent to minimizing, independently for each $i$, the length over bases of $B_i$. Replacing any $\mathcal{C}_i$ by an MCB of $B_i$ strictly decreases the total length while keeping a basis of $\operatorname{ker} A$. Hence an MCB of $G$ is exactly the disjoint union of MCB's of the blocks.
\end{proof}

\end{document}